\documentclass[aps,twocolumn,superscriptaddress,bibnotes,nofootinbib]{revtex4-2}

\usepackage{silence}
\usepackage{amsmath,mathtools,amsthm,amssymb}
\usepackage{algpseudocode}
\usepackage{array}
\usepackage{bbm}
\usepackage{bbold}
\usepackage{comment}
\usepackage{dsfont}
\usepackage{enumitem}
\usepackage{graphicx}
\usepackage{graphics}
\usepackage{mathdots}
\usepackage{mathrsfs}
\usepackage{physics}
\usepackage{pifont}
\usepackage{setspace}
\usepackage{tabularray}
\usepackage{tabularx}
\usepackage{tikz}
\usepackage{xcolor}
\usepackage{wrapfig}

\graphicspath{{figures/}}

\definecolor{elisacolor}{RGB}{230, 0, 90}

\definecolor{myrefcolor}{rgb}{0.067,0.5,0.5}
\definecolor{deepblue}{rgb}{0.1,0.1,0.8}
\definecolor{darkgreen}{rgb}{0.0,0.4,0.0}
\definecolor{richmaroon}{rgb}{0.6,0.0,0.0}
\definecolor{burgundy}{rgb}{0.5,0.0,0.13}
\definecolor{airforceblue}{rgb}{0.36,0.54,0.66}

\tikzset{
  cpdiag color 1/.style={fill=blue!25},
  cpdiag color 2/.style={fill=red!25},
  cpdiag color 3/.style={fill=green!25},
  cpdiag color 4/.style={fill=cyan!30!blue!15},
  cpdiag color 5/.style={fill=yellow!30},
}

\DeclareRobustCommand{\cpnum}[2]{%
  \tikz[baseline=(x.base)]{
    \node[
      #1,
      rounded corners=1pt,
      inner xsep=2pt,
      inner ysep=1pt,
      anchor=base
    ] (x) {$#2$};
  }%
}

\DeclareRobustCommand{\cpdiagdot}[1]{%
  \tikz[baseline=-0.55ex]{
    \node[
      circle,
      draw=black,
      line width=0.25pt,
      #1,
      inner sep=1.6pt
    ] {};
  }%
}

\colorlet{cpdiagwireone}{blue!55}
\colorlet{cpdiagwiretwo}{red!55}
\colorlet{cpdiagwirethree}{green!55}

\usepackage[
    breaklinks,
    colorlinks=true,
    linkcolor=myrefcolor,
    citecolor=myrefcolor,
    urlcolor=myrefcolor
]{hyperref}
\usepackage{orcidlink}

\usepackage[capitalize]{cleveref}

\newtheorem*{theorem*}{Theorem}
\newtheorem*{corollary*}{Corollary}
\newtheorem*{definition*}{Definition}

\newcounter{thm}
\newtheorem{theorem}[thm]{Theorem}
\newtheorem{lemma}[thm]{Lemma}

\newtheorem{proposition}[thm]{Proposition}

\newtheorem{corollary}[thm]{Corollary}

\newtheorem{example}{Example}

\theoremstyle{remark}

\DeclareMathOperator{\sym}{sym}

\newcommand{\HH}{\mathcal{H}}
\newcommand{\BB}{\mathcal{B}}

\newcommand{\proj}[1]{\ket{#1}\!\bra{#1}}
\newcommand{\Comm}{\mathrm{Com}}
\newcommand{\Wg}{\mathrm{Wg}}
\newcommand{\id}{\mathbb{I}}
\newcommand{\E}{\mathbb{E}}
\newcommand{\V}{\mathbb{V}}

\newcommand{\C}{\mathbb{C}}
\newcommand{\Span}{\operatorname{span}}

\newcommand{\be}{\begin{equation}\begin{aligned}\hspace{0pt}}
\newcommand{\ee}{\end{aligned}\end{equation}}
\newcommand{\bbb}{\begin{equation}\begin{aligned}}
\newcommand{\eee}{\end{aligned}\end{equation}}

\allowdisplaybreaks
\makeatletter
\renewcommand{\subsubsection}{\@startsection{subsubsection}{3}{0pt}%
  {1.5ex plus 1ex minus .2ex}%
  {1ex plus .2ex}%
  {\centering\bfseries}}
\makeatother

\begin{document}

\title{Colored Weingarten 
Calculus 
for Block-Unitary Ensemble
}

\author{Daniele Iannotti\orcidlink{0009-0009-0738-5998}}
\email{d.iannotti@ssmeridionale.it}
\affiliation{Scuola Superiore Meridionale, Largo S. Marcellino 10, 80138 Napoli, Italy}
\affiliation{INFN Sezione di Napoli, via Cintia, 80126 Napoli, Italy}
\affiliation{ Institut für Theoretische Physik, Universität zu Köln, Zülpicher Strasse 77, 50937 Köln, Germany}

\author{Elisa Vallini\orcidlink{0009-0003-4613-9526}}
\email{evallini@uni-koeln.de}
\affiliation{ Institut für Theoretische Physik, Universität zu Köln, Zülpicher Strasse 77, 50937 Köln, Germany}

\date{\today}

\begin{abstract}
Haar randomness is the standard null model for quantum typicality, scrambling, and random-matrix formulations of thermalization. In many physical settings, however, randomness is naturally constrained to subspaces of Hilbert space. This structure arises, in particular, both in systems with symmetry-resolved sectors and in the Eigenstate Thermalization Hypothesis (ETH), where mesoscopic energy windows define subspaces in which nearby energy eigenstates are mixed.
Motivated by these settings, we consider a Block-Unitary Ensemble consisting of independently Haar-random unitaries acting on each subspace. We develop a constructive Weingarten Calculus for this ensemble, providing a unified framework across quantum information and many-body physics. This construction reveals a symmetry-resolved Schur--Weyl duality, where the relevant commutant is not the ordinary permutation algebra, but the one of color-preserving permutations.
We apply this framework to quantum states under symmetry constraints and to energy-resolved ensembles for ETH.
\end{abstract}

\maketitle
\section{Introduction}\label{sec:introduction}
When a many-body problem becomes too complicated to solve microscopically, a recurring strategy in physics is to replace detailed microscopic information by a statistical description and ask which properties are typical or universal. This philosophy is central to statistical mechanics and reappears, in a more radical form, in quantum random states, matrices, and unitary ensembles. By deliberately ``mixing'' the microscopic degrees of freedom as thoroughly as possible, one obtains simple reference ensembles from which universal or typical features of otherwise intractable systems can be extracted. 

In quantum information and quantum many-body physics, Haar-randomness provides the canonical realization of this principle. It furnishes a universal reference ensemble for typical entanglement, underlying Page's law, canonical typicality, decoupling, randomized protocols, and random-circuit descriptions of scrambling and quantum chaos~\cite{Page1993,Goldstein2006,Popescu2006,HaydenWinter2008,Dankert2009,HarrowLow2009,BrandaoHarrowHorodecki2016}. 
The full unitary Haar-random ensemble represents the null model when no structure is imposed beyond Hilbert-space dimension; however, physical systems are more restricted.
\begin{figure}
    \centering
    \includegraphics[width=0.9\linewidth]{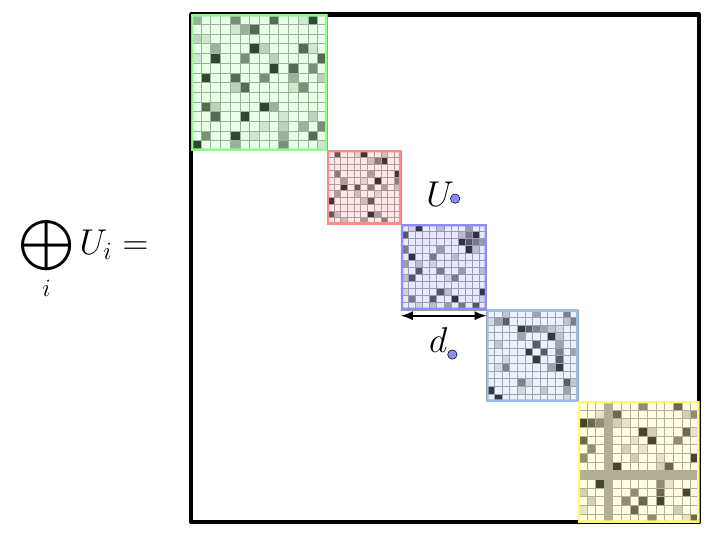}
\caption{Example of a block-unitary matrix $\bigoplus_{i=1}^m U_i$ with $m=5$ blocks, each associated with a distinct color, drawn from the Block-Unitary Ensemble $\mathcal{U}(\mathcal{H}|\boldsymbol M)$. The blocks $U_i$ are independent Haar-random unitary matrices, each of dimension $d_i$.}
    \label{fig:new_cover}
\end{figure}

Conserved quantities constrain the mixing of states across Hilbert space: commuting charges such as energy, particle number, or magnetization define simultaneous symmetry sectors, while non-Abelian symmetries give rise to a richer sector structure due to noncommuting conserved charges.~\cite{YungerHalpern2020Noncommuting}. This point is already central in typicality derivations of thermal states, where equilibrium emerges from a microcanonical subspace specified by sharp conserved quantities, and in resource-theoretic formulations of thermodynamics with constraints~\cite{LostaglioJenningsRudolph2016}. The same sector-resolved structure appears in generalized Page-curve calculations: when the algebra of observables has a nontrivial center, or when states are restricted to fixed symmetry sectors, the usual Haar formula is replaced by a sector-resolved expression whose mean and variance depend on the center and on the dimensions of the charge blocks~\cite{BianchiDona2019}.
Recent extensions to non-Abelian symmetry~\cite{Majidy_2023,Majidy_2023_2} show that this phenomenon is not merely a technical
modification of the Abelian case, but a structural refinement of typical
entanglement in the presence of symmetry~\cite{BianchiDonaKumar2024}.

The same philosophy also underlies statistical formulations of thermalization. The Eigenstate Thermalization Hypothesis (ETH) concerns the matrix elements of physical observables in the energy eigenbasis, whose random structure in chaotic many-body systems underlies the emergence of thermal expectation values~\cite{Deutsch_1991,Srednicki_1999,DAlessio_2016,Deutsch2018,FoiniKurchan2019}.
Importantly, no actual randomness is implied at the microscopic level, rather, chaotic eigenstates and matrix elements of simple observables behave statistically as if they arose from suitable random ensembles. 
Again, this effective randomness cannot be completely structureless. Energy provides a natural constraint: eigenstates far apart in energy cannot be described by the same statistical ensemble, suggesting that randomness should instead emerge locally in energy~\cite{Reimann2018}. This motivates energy-resolved random ensembles, in which nearby eigenstates are randomly mixed within mesoscopic energy windows~\cite{vallini2025refinementseigenstatethermalizationhypothesis}.

\begin{figure}[t]
    \centering
    \includegraphics[width=0.9\linewidth]{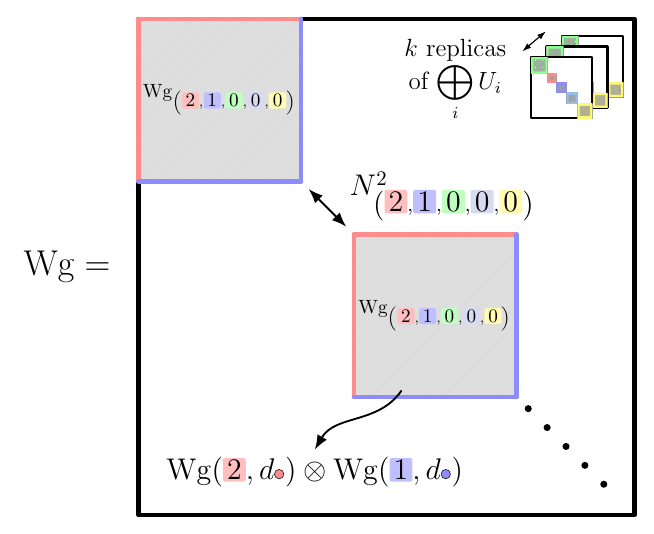}
\caption{Pictorial representation of the block-diagonal structure of the Weingarten matrix for the Block-Unitary Ensemble~\eqref{eq:WC_block_intro} for $k=3$ replicas (top right) and $m=5$ colors. The occupation vector $n =\{n_i\}_{i=1}^m$ specifies how many of the $k$ replicas are associated with each color (block). For each fixed $n$, the Weingarten matrix contains $N_{ n}^2$ equivalent blocks. Each block factorizes into ordinary unitary Weingarten matrices $\mathrm{Wg}(n_i,d_i)$, one for each occupied color. In the figure, we show two of the $N_{ n}^2=9$ blocks for
$n=  \bigl(
        \cpnum{cpdiag color 2}{2},
        \cpnum{cpdiag color 1}{1},
        \cpnum{cpdiag color 3}{0},
        \cpnum{cpdiag color 4}{0},
        \cpnum{cpdiag color 5}{0}
    \bigr) $, corresponding to two red replicas and one blue replica. The corresponding Weingarten block is factorized as $\mathrm{Wg}\!\left(
    \cpnum{cpdiag color 2}{2},
    d_{\cpdiagdot{cpdiag color 2}}
\right)
\otimes
\mathrm{Wg}\!\left(
    \cpnum{cpdiag color 1}{1},
    d_{\cpdiagdot{cpdiag color 1}}
\right)$. The dots on the diagonal indicate the blocks corresponding to all other occupation vectors $n$.}
    \label{fig:Wg}
\end{figure}

Although these windows have a different physical origin from exact symmetry sectors, the two settings share the same mathematical structure, which motivates moving beyond ordinary Haar moments over
$\mathcal{U}(\HH)$ and considering instead moments over the symmetry-compatible unitary group $\mathcal{U}(\HH|\mathbf{M})$, where $\mathbf{M}$ denotes the symmetry operator.
Its elements are block-diagonal in the eigenspaces of $\mathbf{M}$, with each block independently Haar distributed over the corresponding unitary group. We refer to this ensemble as the \emph{Block-Unitary Ensemble}:
\begin{equation}
    U  \in \mathcal{U}(\HH|\mathbf{M})\ , \quad \qquad  U = \bigoplus_{i=1}^m U_i \ .
\end{equation}
An instance is represented in Fig.~\ref{fig:new_cover}; we will refer to a block as a color.

When computing ordinary Haar moments over $\mathcal{U}(\HH)$, one can follow two different routes: either using analytical tools for direct Haar integration or relying on Schur--Weyl duality. The latter relates the collective action of the unitary group on multiple $k$ copies/replicas of $\HH$ to permutations of the copies, providing the structure underlying the Weingarten calculus~\cite{Weingarten1978,CollinsSniady2006,mele_introduction_2024}.
The usefulness of the Weingarten calculus extends far beyond the evaluation of Haar integrals themselves. Haar-moment techniques and their graphical formulations play a central role in the study of random quantum circuits, where they provide a systematic way of averaging over random unitary degrees of freedom~\cite{LiaoGalitski2022,Kostenberger2021,CollinsNechita2010}. This machinery has enabled analytical studies of entanglement growth~\cite{Nahum2017,Nahum2018}, operator spreading and out-of-time-order correlators~\cite{Nahum2018,vonKeyserlingk2018,Leone2021Isospectral,Fritzsch2026, Dowling2025}, spectral statistics in many-body quantum-chaotic systems~\cite{Chan2018,Chan2018Spectral}, as well as Hilbert-space delocalization and participation entropies~\cite{TurkeshiSierant2024}.

This raises the question addressed here: 

\emph{What replaces ordinary Schur--Weyl duality and Weingarten calculus when the Hilbert space is decomposed into subspaces on which independent unitary transformations act?}

In this work, we show that the block structure of the collective unitary action generalizes permutations of the replicas to \emph{color-preserving permutations}. Indeed, each of the $k$ replicas is associated with one of the Hilbert-space blocks, identified with a color, so that a configuration of replicas defines a colored word. The number of replicas of each color is encoded by the occupation vector $ n=\{n_i\}_{i=1}^m$. Permutations act on these colored words and can connect only replicas of the same color.
We denote by $D_{\boldsymbol A}$ the operator associated with a color-preserving permutation, indexed by $\boldsymbol A\in\mathcal C_k$. Then, the Block-Unitary Ensemble average of an operator $\mathcal{O}$ on $k$ copies of $\mathcal{H}$ is
\begin{equation}
\begin{split}
    \Phi_{\mathcal{U}(\mathcal{H}|\boldsymbol M)}^{(k)}(\mathcal{O})
    = \sum_{\boldsymbol A,\boldsymbol B\in\mathcal C_k}
    \Wg_{\boldsymbol A,\boldsymbol B}
    \Tr(D_{\boldsymbol A}^{\dagger}\mathcal O)
    D_{\boldsymbol B} \ .
\end{split}
\label{eq:WC_intro}
\end{equation}

Our main result is that the corresponding \emph{Weingarten matrix inherits a precise block-diagonal structure}. The blocks are labeled by the occupation vector $n$. For a fixed $n$, there are
$N_{n}=k!/\prod_{i=1}^{m}n_i!$ different colored words with the same occupation numbers, giving rise to $N_{n}^2$ equivalent blocks in the Weingarten matrix. Within each of these blocks, the Weingarten matrix factorizes into a tensor product of ordinary Weingarten matrices $\mathrm{Wg}(n_i,d_i)$ of the unitary group, one for each color. The full Weingarten matrix therefore takes the form
\begin{equation}
    \Wg
    = \bigoplus_{ n}
    \left(
    \mathbb{I}_{N_{n}^2}
    \otimes
    \bigotimes_{i=1}^{m} \mathrm{Wg}(n_i,d_i)
    \right) \ .
\label{eq:WC_block_intro}
\end{equation}
This structure is illustrated in Fig.~\ref{fig:Wg}.

The rest of the paper explains how this structure arises and develops its consequences. We begin in Sec.~\ref{sec_comm} by constructing the commutant of the collective
block-unitary action, which determines the elements $D_{\boldsymbol A}$. A characterization of this commutant was initiated by Hearth et al.~\cite{Hearth_2025}. Here, we provide an alternative constructive derivation in terms of color-preserving permutations, which was proved to be a group algebra for blocks of equal size~\cite{foxman2026quantumlazysamplingpath}.
This construction allows us to formulate the corresponding Weingarten
calculus explicitly in Sec.~\ref{sec_Weing}, turning block-resolved Haar
integration into a practical moment calculus. We then illustrate the broader scope of the Block-Unitary Ensemble through applications to random quantum states in Sec.~\ref{sec_quantumstates} and to energy-resolved random ensembles in the context of ETH in Sec.~\ref{sec_ETH}. We conclude with a discussion and outlook in Sec.~\ref{sec_conclusion}.

\section{Block-Unitary Ensemble's Commutant}
\label{sec_comm}
A convenient way to organize Haar averages is to first identify the commutant of the collective unitary action. The operators in this commutant are precisely the invariant structures that survive Haar averaging and therefore provide the building blocks for computing its moments.

One can define the $k$-commutant of a subgroup of the unitary group $\mathcal{G}\subset \mathcal{U}(\HH)$ as the set of bounded operators on $k$ replicas of $\HH$, namely
\begin{equation}
    \Comm_k(\mathcal{G})
    :=
    \{X\in\BB(\HH^{\otimes k})|[X,U^{\otimes k}]=0\,, \forall U\in \mathcal{G}\} .
    \label{eq:commutant_def}
\end{equation}
Thus, $\Comm_k(\mathcal{G})$ is the space of operators on $k$ replicas that are invariant under the diagonal action of the ensemble, and it is the image of the twirling map
\begin{equation}
\Phi_{\mathcal{G}}^{(k)}(\mathcal{O}):=\int_{\mathcal{G}} dU \,U^{\otimes k}\mathcal{O}(U^\dagger)^{\otimes k} \ ,
\end{equation}
where $dU$ denotes Haar measure for continuous groups and the uniform discrete measure for finite groups, while $\mathcal{O}$ is a bounded operator in $\mathcal{H}^{\otimes k}$. In fact, the map $\Phi_{\mathcal{G}}^{(k)}$ is the Hilbert--Schmidt projector onto $\Comm_k(\mathcal{G})$.

In this work, we study the commutant structure of the Block-Unitary Ensemble.
Let the Hilbert space be decomposed as
\begin{equation}
    \HH=\bigoplus_{i=1}^m \HH_i \ ,
\end{equation}
with $d=\dim\HH$, $d_i=\dim\HH_i$. For later use, we denote the projectors on each block as $\Pi_i:\HH\to\HH_i$.

We define the Block-Unitary Ensemble as
\begin{equation}
    \mathcal{U}(\mathcal{H}|\boldsymbol M)
    =
    \mathcal{U}(\HH_1)\times\cdots\times\mathcal{U}(\HH_m) \ ,
    \label{eq:blockunitary}
\end{equation}
with unitary representation
$\rho:\mathcal{U}(\mathcal{H}|\boldsymbol M)\to\mathcal{U}(\HH)$ defined by
\begin{equation}
    \rho(U_1,\ldots,U_m)
    =
    \bigoplus_{i=1}^m U_i \ .
    \label{eq:embeddingsum}
\end{equation}
We henceforth identify $\mathcal{U}(\mathcal{H}|\boldsymbol M)$ with its image
$\rho(\mathcal{U}(\mathcal{H}|\boldsymbol M))\subseteq\mathcal{U}(\HH)$, which is a compact subgroup of the unitary group.

In order to properly define the commutant of $\mathcal{U}(\HH|\textbf{M})$ as generated by color-preserving permutations, we first define them.

\subsection{Color-preserving permutations}
We introduce the following notations.\\
We define as \emph{colored word} the vector 
\begin{equation}
    \boldsymbol{\alpha}=(\alpha_1,\dots,\alpha_k)\in\{1,\dots,m\}^k  \ ,
\end{equation} indicating the color chosen for each replica and as \emph{colored projectors} 
\begin{equation}
    \Pi_{\boldsymbol{\alpha}}
    :=
    \Pi_{\alpha_1}\otimes\cdots\otimes\Pi_{\alpha_k} \ ,
    \label{eq:word_projector}
\end{equation}
such that $\HH_{\boldsymbol{\alpha}}:=\Pi_{\boldsymbol{\alpha}}\HH^{\otimes k} \subseteq \HH^{\otimes k}$.\\
We can now define the \emph{occupation vector} of a certain colored word as
\begin{equation}
    n(\boldsymbol{\alpha})=(n_1(\boldsymbol{\alpha}),\dots,n_m(\boldsymbol{\alpha}))\ ,
\end{equation}
where $n_i(\boldsymbol{\alpha}):=\#\{r:\alpha_r=i\}$ indicates in how many replicas the color $i$ is chosen, with condition $\sum_i n_i = k$. The occupation numbers depend only on the number of replicas of each color, and not on their positions within the colored word. \\ 
Consider as an example for $k=5$ and $m\geq 3$,
\begin{equation}
    \boldsymbol{\alpha}
=
\bigl(
\cpdiagdot{cpdiag color 2},
\cpdiagdot{cpdiag color 1},
\cpdiagdot{cpdiag color 3},
\cpdiagdot{cpdiag color 2},
\cpdiagdot{cpdiag color 1}
\bigr)
\quad \alpha_1,\alpha_4= \cpdiagdot{cpdiag color 2}, \alpha_2,\alpha_5=\cpdiagdot{cpdiag color 1},
\alpha_3=\cpdiagdot{cpdiag color 3} \ .
\label{eq:ex_alpha}
\end{equation}
and 
\begin{equation}
    n\bigl(
        \cpdiagdot{cpdiag color 2},
        \cpdiagdot{cpdiag color 1},
        \cpdiagdot{cpdiag color 3},
        \cpdiagdot{cpdiag color 2},
        \cpdiagdot{cpdiag color 1}
    \bigr)
    =
    \bigl(
        \cpnum{cpdiag color 2}{2},
        \cpnum{cpdiag color 1}{2},
        \cpnum{cpdiag color 3}{1},
        \cpnum{cpdiag color 4}{0},
        \cdots,
        \cpnum{cpdiag color 5}{0}
    \bigr) \ .
\end{equation}
With this relabeling, we can decompose the $k$-fold Hilbert space in
\begin{equation}
    \HH^{\otimes k}
    \simeq
    \bigoplus_{\boldsymbol{\alpha}\in\{1,\dots,m\}^k}\HH_{\boldsymbol{\alpha}}
    \simeq
    \bigoplus_{\substack{n\in\mathbb{N}^m \\ |n|=k}}\HH_n\ ,
    \qquad
    \label{eq:occupation_decomp}
\end{equation}
with $\HH_n:=\bigoplus_{\boldsymbol{\alpha}:\,n(\boldsymbol{\alpha})=n}\HH_{\boldsymbol{\alpha}}$, which follows from the existence of a canonical isomorphism for finite $m$, with $\simeq$ denoting an isomorphism between the two spaces.
Since $U_i$ acts only on $\HH_i$, each $\HH_{\boldsymbol{\alpha}}$ is invariant under
$U^{\otimes k}$.

Let $S_k$ be the permutation group of $k$ elements that acts on $\HH^{\otimes k}$ by
\begin{equation}
    \V_\sigma
    \ket{x_1,\dots,x_k}
    =
    \ket{x_{\sigma^{-1}(1)},\dots,x_{\sigma^{-1}(k)}}\ ,
    \label{eq:permutation_operator}
\end{equation}
for $\ket{x_1},\dots,\ket{x_k} \in \C^d$. \\ 

Given a colored word $\boldsymbol{\beta}$,
a permutation $\sigma\in S_k$ acts on its entries to produce another
colored word $\boldsymbol{\alpha}$. We call $\sigma$ \emph{color-preserving}
if
\begin{equation}
   \alpha_r = \beta_{\sigma^{-1}(r)}\ ,
    \qquad r=1,\ldots,k \ .
\end{equation} 
With this setting one can prove the following theorem, where the generalized Schur--Weyl duality characterizes the commutant as the algebra generated by tensor products of sector projectors and permutations of the replicas~\cite{MarvianSpekkens2014}.

\begin{figure}
    \centering
    \includegraphics[width=0.9\linewidth]{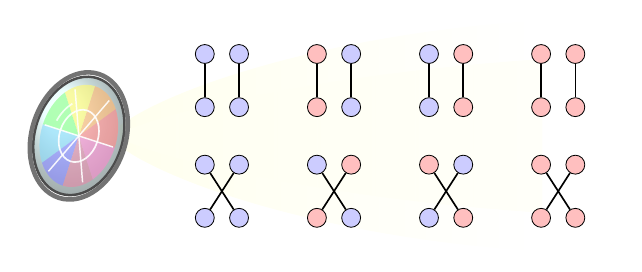}
    \caption{Diagrammatic representation of the elements of $\Comm_k(\mathcal{U}(\HH|\textbf{M}))$ for $k,m=2$. Each color-preserving permutation graph between two colored words is such that there are $k$ columns and $m$ possible colors for each filled circle.}
    \label{fig:icon_paper}
\end{figure}

\begin{theorem}[Commutant of the Block-Unitary Ensemble]
\label{thm:block_commutant}
Let $\mathcal{G}=\mathcal{U}(\HH|\textbf{M})$ act block diagonally on
$\HH=\bigoplus_i\HH_i$. Then the commutant is the span of color-preserving permutation diagrams between colored words
\begin{equation}
\small
    \Comm_k(\mathcal{G})
    =
    \Span\left\{
       \Pi_{\boldsymbol{\alpha}}\V_\sigma\Pi_{\boldsymbol{\beta}}:
        \sigma\in S_k,\,
        \alpha_r=\beta_{\sigma^{-1}(r)}\,
    \forall r
    \right\}\ ,
    \label{eq:sr_schur_weyl}
\end{equation}
with the colored projectors $\Pi_{\boldsymbol{\alpha}}$ and $\Pi_{\boldsymbol{\beta}}$ defined in~\eqref{eq:word_projector}, and the permutation operator $\V_\sigma$ in~\eqref{eq:permutation_operator}.
Equivalently,
\begin{equation}
\small
    \Comm_k(\mathcal{G})
    =
    \bigoplus_{\substack{n\in\mathbb{N}^m\\ |n|=k}}
    \Span
    \left\{
        \Pi_{\boldsymbol{\alpha}}\V_\sigma\Pi_{\boldsymbol{\beta}}:
        n(\boldsymbol{\alpha})=n(\boldsymbol{\beta})=n
    \right\}\ .
\end{equation}
\end{theorem}
The dimension of the commutant for $d_i\geq k$ is
\begin{equation}
    \dim \Comm_k(\mathcal{U}(\HH|\textbf{M})) = k!\,m^k,
\end{equation}
since there are $m^k$ possible colored words $\boldsymbol{\alpha}$ and
$k!$ permutations for each word.\\

Graphically, each element of the commutant is represented
by two rows of $k$ colored vertices, as shown in Fig.~\ref{fig:icon_paper} for $k,m=2$. The upper row carries the output colored word $\boldsymbol{\alpha}$, while the lower row carries the input colored word $\boldsymbol{\beta}$. The permutation $\sigma$ is represented by connecting the lower vertex $r$ to the upper vertex $\sigma(r)$ for every $r$. The color-preserving condition simply states that every wire connects vertices of the same color.
If this condition is not satisfied, the corresponding operator vanishes.

We illustrate how the block structure modifies the ordinary permutation algebra through the minimal case of two copies. We first consider the standard Haar setting and then introduce a decomposition into two blocks.

\begin{example}[Two copies]
\label{example:two_qubits}
\begingroup
\newcommand{\cblue}{\cpdiagdot{cpdiag color 1}}
\newcommand{\cred}{\cpdiagdot{cpdiag color 2}}

Consider first the ordinary Haar action on two qubits. The singlet state
\begin{equation}
    \ket{\psi}
    :=
    \frac{\ket{01}-\ket{10}}{\sqrt{2}}
\end{equation}
can be characterized equivalently through the collective unitary action
and the permutation action. For any $U\in\mathcal U(\C^2)$,
\begin{equation}
    (U\otimes U)\ket{\psi}
    =
    \det(U)\ket{\psi}\ ,
    \label{eq:det}
\end{equation}
while under the swap $s=(12)$ and the identity $e=(1)(2)$
\begin{equation}
    \V_s\ket{\psi}
    =
    -\ket{\psi}\,, \quad  \V_e\ket{\psi}
    =
    +\ket{\psi}\ .
    \label{eq:s_e_2qubit}
\end{equation}
This is the simplest manifestation of the usual Schur--Weyl duality:
the collective Haar-unitary action and the permutation action provide
two complementary descriptions of the invariant decomposition of
$(\C^2)^{\otimes 2}$. \\ 

We now contrast this with the Block-Unitary Ensemble.
Let $\HH=\C^2\otimes \C^2=\HH_{\cblue}\oplus \HH_{\cred}$ such that
\begin{equation}
    \HH_{\cblue}=\Span\{\ket{00},\ket{01}\},\,
    \HH_{\cred}=\Span\{\ket{10},\ket{11}\}.
\end{equation}
A block-unitary has the form
\begin{equation}
    U=
    \ket{0}\bra{0}\otimes U_{\cblue}
    +
    \ket{1}\bra{1}\otimes U_{\cred}
    =\begin{pmatrix}
U_{\cblue} & 0 \\
0 & U_{\cred}
\end{pmatrix},
\end{equation}
and $U_{\cblue},U_{\cred}\in \mathcal{U}(\mathbb{C}^2)$.
Consider the two-replica state
\begin{equation}
    \ket{\psi_{\cblue}}
    :=
    \frac{
    \ket{00}\otimes\ket{01}
    -
    \ket{01}\otimes\ket{00}
    }{\sqrt 2}
    \in
    \HH_{\cblue}\otimes\HH_{\cblue}\ ,
\end{equation}
since both replicas lie in the $\cblue$ sector,
\begin{equation}
    (U\otimes U)\ket{\psi_{\cblue}}
    =
    \det(U_{\cblue})\ket{\psi_{\cblue}} \ .
\end{equation}
On the other hand, for the replica swap with $s=(12)$,
\begin{equation}
    \V_s^{\cblue}\ket{\psi_{\cblue}}
    := \Pi_{(\cblue,\cblue)}\V_s\Pi_{(\cblue,\cblue)}
    \ket{\psi_{\cblue}}=
    -\ket{\psi_{\cblue}} \ .
\end{equation}
The same holds for $\HH_{\cred}\otimes\HH_{\cred}$.\\
Consider now a mixed-color two-replica state. For $a,b\in\{0,1\}$, define a state in $\left(\HH_{\cblue}\otimes\HH_{\cred}\right) \oplus\left(\HH_{\cred}\otimes\HH_{\cblue}\right)$
\begin{equation}
    \ket{\psi_{\cblue\cred}(a,b)}
    :=
    \frac{
    \ket{0a}\otimes\ket{1b}
    -
    \ket{1b}\otimes\ket{0a}
    }{\sqrt 2}\ .
\end{equation}
One can then ask if this state is an eigenstate of two operators, similarly to the two-qubit case for the swap and the identity in Eq.~\eqref{eq:s_e_2qubit}.
The identity can be written as
\begin{equation}
    \V_e^{\cblue\cred}
    :=
    \Pi_{(\cred,\cblue)}\V_e\Pi_{(\cred,\cblue)}+\Pi_{(\cblue,\cred)}\V_e\Pi_{(\cred,\cblue)}\,,
\end{equation}
whereas the corresponding color-resolved swap is
\begin{equation}
    \V_s^{\cblue\cred}
    :=
    \Pi_{(\cred,\cblue)}\V_s\Pi_{(\cblue,\cred)}
    +
    \Pi_{(\cblue,\cred)}\V_s\Pi_{(\cred,\cblue)}\  ,
\end{equation}
which acts as
\begin{equation}
    \V_s^{\cblue\cred}
    \ket{\psi_{\cblue\cred}(a,b)}
    =
    -\ket{\psi_{\cblue\cred}(a,b)} \ .
\end{equation}
Under the block-unitary action $U\otimes U$, the mixed-color antisymmetric states generally mix among themselves, while their span remains invariant. To see this, first note that the action on a single copy is
\begin{equation}
\begin{aligned}
U\ket{0a} &=\sum_{a'=0}^1 (U_{\cblue})_{a'a}\ket{0a'},\\
U\ket{1b} &=\sum_{b'=0}^1 (U_{\cred})_{b'b}\ket{1b'},
\end{aligned}
\end{equation}
and as we did in Eq.~\eqref{eq:det} with the state $\ket{\psi_{\cblue\cred}(a,b)}$, we have
\begin{equation}
(U\otimes U)
\ket{\psi_{\cblue\cred}(a,b)}
=
\sum_{a',b'=0}^1
(U_{\cblue})_{a'a}
(U_{\cred})_{b'b}
\ket{\psi_{\cblue\cred}(a',b')}.
\end{equation}
Thus, $(U\otimes U)\mathcal A^-_{\cblue\cred}\subseteq\mathcal A^-_{\cblue\cred}$, where
$\mathcal A^-_{\cblue\cred}:=\Span\{\ket{\psi_{\cblue\cred}(a,b)}:a,b\in\{0,1\}\}$.
Unlike the same-color singlet, an individual mixed-color antisymmetric state is generally not mapped to a scalar multiple of itself.

\endgroup
\end{example}

\section{Symmetry-resolved Weingarten calculus}
\label{sec_Weing}
The above  result allows us to express the moment operator in terms of the so-called symmetry-resolved Weingarten coefficients
\begin{equation}
\begin{split}
    \Phi_{\mathcal{U}(\HH|\textbf{M})}^{(k)}(\mathcal{O})&= \sum_{\boldsymbol{A},\boldsymbol{B} \in \mathcal{C}_k} \Wg_{\boldsymbol{A},\boldsymbol{B}} \Tr(D_{\boldsymbol{A}}^\dagger \mathcal{O}) D_{\boldsymbol{B}} \ ,
\end{split}
\label{eq:WC}
\end{equation}
with 
\begin{equation}
    D_{\boldsymbol{A}}
    :=
    \Pi_{\boldsymbol{\alpha}}
    \V_\sigma
    \Pi_{\boldsymbol{\beta}}
    \in \Comm_k(\mathcal{U}(\HH|\textbf{M}))\ ,
    \quad
    \boldsymbol{A}
    =
    (\boldsymbol{\alpha},\boldsymbol{\beta};\sigma) \ ,
\end{equation} 
and let $\mathcal{C}_k$ denote the set of color-preserving permutation triples indeces.

The symmetry-resolved Weingarten coefficients are determined by the overlaps
between the commutant basis elements, encoded in the Gram matrix
\begin{equation}
G_{\boldsymbol{A},\boldsymbol{B}} :=\Tr\!\left(D_{\boldsymbol{A}}^\dagger D_{\boldsymbol{B}}\right).
\end{equation}
The corresponding Weingarten matrix in Eq.~\eqref{eq:WC} is then given by the Moore--Penrose pseudoinverse when the commutant basis is linearly dependent. In the regime in which all Gram blocks are invertible, the pseudoinverse reduces to the ordinary inverse. Our next goal is to give both the Gram matrix and the resulting Weingarten coefficients an explicit combinatorial and graphical representation in terms
of color-preserving permutation diagrams.

\subsection{Gram Matrix}
\label{sec_Gram}
In order to compute the Gram matrix we introduce the following notation.
For a colored word $\boldsymbol{\alpha}$ 
and for each color $i=1,\ldots,m$, we define the set of positions occupied by color $i$ as
\begin{equation}
    P_i(\boldsymbol{\alpha})
    :=
    \{r\in\{1,\ldots,k\}:\alpha_r=i\}\ .
\end{equation}
The color-preserving condition implies that a permutation $\sigma$ induces, for every color $i$, a bijection
\begin{equation}
    \sigma_i
    :=
    \left.\sigma\right|_{P_i(\boldsymbol{\beta})}
    :
    P_i(\boldsymbol{\beta})
    \longrightarrow
    P_i(\boldsymbol{\alpha}) \ ,
    \label{eq:colour_restricted_sigma}
\end{equation}
and similarly
\begin{equation}
    \tau_i
    :=
    \left.\tau\right|_{P_i(\boldsymbol{\delta})}
    :
    P_i(\boldsymbol{\delta})
    \longrightarrow
    P_i(\boldsymbol{\gamma}) \ .
\end{equation}

Whenever $\boldsymbol{\alpha}=\boldsymbol{\gamma}$ and $\boldsymbol{\beta}=\boldsymbol{\delta}$, the two restricted maps have the same domain and codomain, so that
\begin{equation}
    \sigma_i^{-1}\tau_i
    :
    P_i(\boldsymbol{\beta})
    \longrightarrow
    P_i(\boldsymbol{\beta})
\end{equation}
is a permutation. Moreover $\sigma_i^{-1}\tau_i = \left.(\sigma^{-1}\tau)\right|_{P_i(\boldsymbol{\beta})}$.

Let us consider the example of Eq.~\eqref{eq:ex_alpha}, with $k=5$ and $m\geq3$. We fix the two colored words $\boldsymbol{\alpha} = \bigl(
\cpdiagdot{cpdiag color 2},
\cpdiagdot{cpdiag color 1},
\cpdiagdot{cpdiag color 3},
\cpdiagdot{cpdiag color 2},
\cpdiagdot{cpdiag color 1}
\bigr)$ and $\boldsymbol{\beta} = \bigl(
\cpdiagdot{cpdiag color 3},
\cpdiagdot{cpdiag color 1},
\cpdiagdot{cpdiag color 2},
\cpdiagdot{cpdiag color 1},
\cpdiagdot{cpdiag color 2}
\bigr)$.
The corresponding position sets are
\begin{align}
P_{\cpdiagdot{cpdiag color 1}}(\boldsymbol{\alpha})
    &=\{2,5\},
&
P_{\cpdiagdot{cpdiag color 1}}(\boldsymbol{\beta})
    &=\{2,4\},
\nonumber\\
P_{\cpdiagdot{cpdiag color 2}}(\boldsymbol{\alpha})
    &=\{1,4\},
&
P_{\cpdiagdot{cpdiag color 2}}(\boldsymbol{\beta})
    &=\{3,5\},
\nonumber\\
P_{\cpdiagdot{cpdiag color 3}}(\boldsymbol{\alpha})
    &=\{3\},
&
P_{\cpdiagdot{cpdiag color 3}}(\boldsymbol{\beta})
    &=\{1\}.
\label{app:eq:position_sets_example}
\end{align}

We consider two color-preserving permutations,
\begin{equation}
    \sigma=(13)(254),
    \qquad
    \tau=(1345)(2),
\label{app:eq:sigma_tau_example}
\end{equation}
both mapping the boundary word $\boldsymbol{\beta}$ to
$\boldsymbol{\alpha}$. Pictorially,
\begin{equation}
\sigma =
\vcenter{\hbox{%
\begin{tikzpicture}[
    x=0.62cm,
    y=0.72cm,
    wire/.style={
        line width=1.5pt,
        line cap=round
    }
]


\draw[wire, draw=cpdiagwirethree]
(2,1) -- (0,0);

\draw[wire, draw=cpdiagwireone]
(4,1) -- (1,0);

\draw[wire, draw=cpdiagwireone]
(1,1) -- (3,0);

\draw[wire, draw=cpdiagwiretwo]
(0,1) -- (2,0);

\draw[wire, draw=cpdiagwiretwo]
(3,1) -- (4,0);

\node at (0,1) {\cpdiagdot{cpdiag color 2}};
\node at (1,1) {\cpdiagdot{cpdiag color 1}};
\node at (2,1) {\cpdiagdot{cpdiag color 3}};
\node at (3,1) {\cpdiagdot{cpdiag color 2}};
\node at (4,1) {\cpdiagdot{cpdiag color 1}};

\node at (0,0) {\cpdiagdot{cpdiag color 3}};
\node at (1,0) {\cpdiagdot{cpdiag color 1}};
\node at (2,0) {\cpdiagdot{cpdiag color 2}};
\node at (3,0) {\cpdiagdot{cpdiag color 1}};
\node at (4,0) {\cpdiagdot{cpdiag color 2}};

\end{tikzpicture}%
}},
\label{eq:sigma_diagram}
\end{equation}

\begin{equation}
\tau =
\vcenter{\hbox{%
\begin{tikzpicture}[
    x=0.62cm,
    y=0.72cm,
    wire/.style={
        line width=1.5pt,
        line cap=round
    }
]


\draw[wire, draw=cpdiagwirethree]
(2,1) -- (0,0);

\draw[wire, draw=cpdiagwireone]
(1,1) -- (1,0);

\draw[wire, draw=cpdiagwireone]
(4,1) -- (3,0);

\draw[wire, draw=cpdiagwiretwo]
(3,1) -- (2,0);

\draw[wire, draw=cpdiagwiretwo]
(0,1) -- (4,0);

\node at (0,1) {\cpdiagdot{cpdiag color 2}};
\node at (1,1) {\cpdiagdot{cpdiag color 1}};
\node at (2,1) {\cpdiagdot{cpdiag color 3}};
\node at (3,1) {\cpdiagdot{cpdiag color 2}};
\node at (4,1) {\cpdiagdot{cpdiag color 1}};

\node at (0,0) {\cpdiagdot{cpdiag color 3}};
\node at (1,0) {\cpdiagdot{cpdiag color 1}};
\node at (2,0) {\cpdiagdot{cpdiag color 2}};
\node at (3,0) {\cpdiagdot{cpdiag color 1}};
\node at (4,0) {\cpdiagdot{cpdiag color 2}};

\end{tikzpicture}%
}}.
\label{eq:tau_diagram}
\end{equation}

Their restrictions to each color sector are the bijections $\sigma_i$ and $\tau_i$, which are graphically identified with the wires of color $i$.

We now consider the permutations $\sigma_i^{-1}\tau_i$. These can be obtained either by composing the restricted maps $\sigma_i^{-1}$ and $\tau_i$, or equivalently by first constructing the global relative permutation $\sigma^{-1}\tau$ and then restricting it to the corresponding color-position set. The two procedures give the same result. In the present example, this yields

\begin{equation}
\sigma^{-1}\tau =
\vcenter{\hbox{%
\begin{tikzpicture}[
    x=0.62cm,
    y=0.72cm,
    wire/.style={
        line width=1.5pt,
        line cap=round
    }
]


\draw[wire, draw=cpdiagwirethree]
(0,1) -- (0,0);

\draw[wire, draw=cpdiagwireone]
(1,1) -- (3,0);

\draw[wire, draw=cpdiagwireone]
(3,1) -- (1,0);

\draw[wire, draw=cpdiagwiretwo]
(2,1) -- (4,0);

\draw[wire, draw=cpdiagwiretwo]
(4,1) -- (2,0);

\node at (0,1) {\cpdiagdot{cpdiag color 3}};
\node at (1,1) {\cpdiagdot{cpdiag color 1}};
\node at (2,1) {\cpdiagdot{cpdiag color 2}};
\node at (3,1) {\cpdiagdot{cpdiag color 1}};
\node at (4,1) {\cpdiagdot{cpdiag color 2}};

\node at (0,0) {\cpdiagdot{cpdiag color 3}};
\node at (1,0) {\cpdiagdot{cpdiag color 1}};
\node at (2,0) {\cpdiagdot{cpdiag color 2}};
\node at (3,0) {\cpdiagdot{cpdiag color 1}};
\node at (4,0) {\cpdiagdot{cpdiag color 2}};

\end{tikzpicture}%
}},
\label{eq:sigma_inverse_tau_diagram}
\end{equation}

so that 
\begin{align}
    \sigma^{-1}_{\cpdiagdot{cpdiag color 1}} 
    \tau_{\cpdiagdot{cpdiag color 1}} = s \quad\quad\quad
\sigma^{-1}_{\cpdiagdot{cpdiag color 2}} 
    \tau_{\cpdiagdot{cpdiag color 2}} = s \quad\quad\quad
    \sigma^{-1}_{\cpdiagdot{cpdiag color 3}} 
    \tau_{\cpdiagdot{cpdiag color 3}} = e \ .
\end{align}

We are then ready to state the Gram matrix expression.

\begin{proposition}[Gram matrix]
\label{prop:block_commutant_gram}
Let $\boldsymbol{A}=(\boldsymbol{\alpha},\boldsymbol{\beta};\sigma)$ and $\boldsymbol{B}=(\boldsymbol{\gamma},\boldsymbol{\delta};\tau)$ be two nonzero color-preserving diagrams, with $D_{\boldsymbol{A}}:=\Pi_{\boldsymbol{\alpha}}\V_\sigma\Pi_{\boldsymbol{\beta}}, \, D_{\boldsymbol{B}}:=\Pi_{\boldsymbol{\gamma}}\V_\tau\Pi_{\boldsymbol{\delta}}$. The Hilbert--Schmidt Gram matrix is then
\begin{equation}
    G_{\boldsymbol{A},\boldsymbol{B}}
    :=
    \Tr\!\left(D_{\boldsymbol{A}}^\dagger D_{\boldsymbol{B}}\right)
    =
    \delta_{\boldsymbol{\alpha},\boldsymbol{\gamma}}
    \delta_{\boldsymbol{\beta},\boldsymbol{\delta}}
    \prod_{i=1}^{m}
    d_i^{\#\operatorname{cycles}(\sigma_i^{-1}\tau_i)}\ .
    \label{eq:block_commutant_gram_cycles}
\end{equation}
\end{proposition}

The Hilbert--Schmidt product $G_{\boldsymbol A,\boldsymbol B}$ of two diagrams is computed by stacking the
reflected first diagram on top of the second one.
The product is zero unless the two diagrams have the same boundary words $\boldsymbol{\alpha}=\boldsymbol{\gamma}$ and $\boldsymbol{\beta}=\boldsymbol{\delta}$.
When the boundary words agree, the stacked diagram decomposes into closed
loops. A closed loop of color $i$ contributes a factor $d_i$.\\
The restricted maps $\sigma_i$ and $\tau_i$ can also be identified with permutations in $S_{n_i}$ by order-preserving relabelings of their respective domain and codomain. Denoting the resulting permutations by $\widehat{\sigma}_i,\widehat{\tau}_i\in S_{n_i}$, the cycle structure is preserved under this identification, so that
\begin{equation}
d_i^{\#\operatorname{cycles}(\sigma_i^{-1}\tau_i)}
=
G(n_i,d_i)_{\widehat{\sigma}_i,\widehat{\tau}_i}.
\label{eq:restricted_gram_identification}
\end{equation}
The relabeling and its relation to the restricted maps are made explicit in Appendix~\ref{app:C}.

In the example above the boundary words $\boldsymbol{\alpha}$ and $\boldsymbol{\beta}$ are fixed, and we have selected two particular color-preserving permutations $\sigma$ and $\tau$. We are therefore computing a single matrix element of the Gram block associated with these boundary words. For this pair,
\begin{equation}
G_{(\boldsymbol{\alpha},\boldsymbol{\beta};\sigma),
  (\boldsymbol{\alpha},\boldsymbol{\beta};\tau)}
=
\prod_{i=1}^{m}
d_i^{\#\operatorname{cycles}(\sigma_i^{-1}\tau_i)}
=
d_{\cpdiagdot{cpdiag color 1}}\,
d_{\cpdiagdot{cpdiag color 2}}\,
d_{\cpdiagdot{cpdiag color 3}}.
\label{eq:app_example_gram_element}
\end{equation}

\begin{figure}[t]
    \centering
    \includegraphics[width=0.98\linewidth]{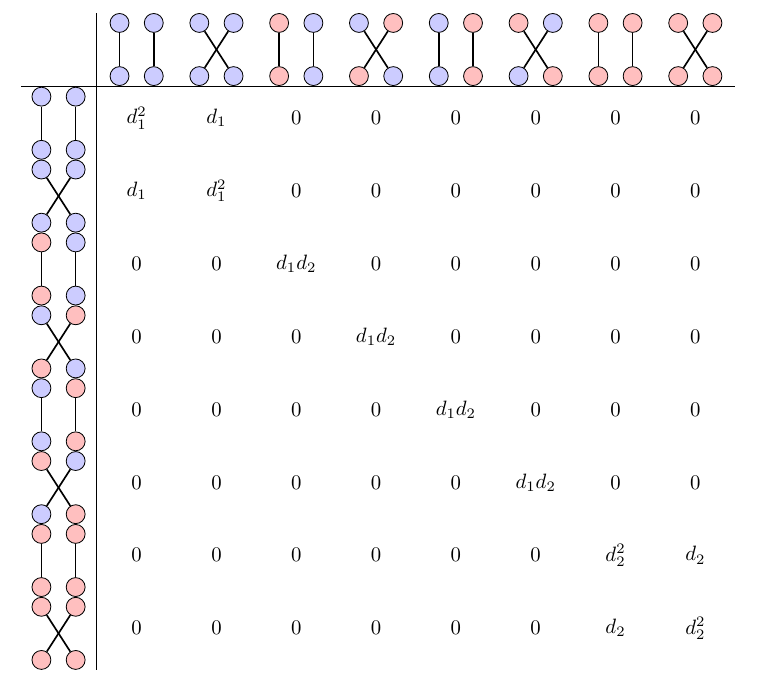}
    \caption{Gram matrix of the color-preserving permutation operators for two replicas and two sectors ($k=m=2$), with sector dimensions $d_1$ and $d_2$. Rows and columns are labeled by diagrams $D_{\boldsymbol A}$, and each entry is the Hilbert--Schmidt overlap $G_{\boldsymbol A,\boldsymbol B}=\Tr(D_{\boldsymbol A}^{\dagger}D_{\boldsymbol B})$. Graphically, this overlap is evaluated by composing the adjoint of the first diagram with the second and closing the remaining wires. Each closed loop of color $i$ contributes a factor $d_i$, while incompatible colors give zero. The same-color diagrams form the two $2\times2$ blocks $G(2,d_1)$ and $G(2,d_2)$; the four mixed-color diagrams are then $d_1d_2$.}
    \label{fig:gram_k2m2}
\end{figure}

The previous proposition immediately yields the following block decomposition
of the full Gram matrix.

\begin{corollary}[Block decomposition of the Gram matrix]
\label{cor:gram_decomposition}
The full Gram matrix decomposes as
\begin{align}
    G
    = &
    \bigoplus_{\substack{n,\\ |n|=k}} 
    \left[
        \bigotimes_{i=1}^{m} G(n_i,d_i)
    \right]^{\oplus N_{n}^2} \\
    \equiv &
    \bigoplus_{ n}
    \left(
    \mathbb{I}_{N_{ n}^2}
    \otimes
    \bigotimes_{i=1}^{m} G(n_i,d_i)
    \right)\ , 
    \\ 
    &\mathrm{with }\quad N_{n}
    :=
    \frac{k!}{\prod_{i=1}^{m} n_i!}\ ,
\label{eq:graphical_loop_rule}
\end{align}
where $G(n_i,d_i)$ is the ordinary unitary Gram matrix on $n_i$ replicas in dimension $d_i$.
\end{corollary}
The block structure in Eq.~\eqref{eq:graphical_loop_rule} becomes manifest once the color-preserving permutation basis is suitably ordered. We first group the basis elements according to their occupation vector $n$. For a fixed $n$, the number of colored words with these occupation numbers is $N_n=k!/\prod_{i=1}^m n_i!$,
since this is the number of permutations of $k$ elements with repetitions, where color $i$ appears $n_i$ times. We then group the basis elements according to the pair of boundary words
$(\boldsymbol{\alpha},\boldsymbol{\beta})$ compatible with the fixed $n$. Since both words must have the same occupation vector $n$, there are $N_n^2$ such pairs. For each fixed pair $(\boldsymbol{\alpha},\boldsymbol{\beta})$, the remaining degrees of freedom are independent permutations within each color, giving rise to the tensor product $\bigotimes_i G(n_i,d_i)$. Hence, for each occupation vector $n$, this tensor-product block appears with multiplicity $N_n^2$.

In the example above the explicit pair $\sigma$ and $\tau$ gives a single matrix element of the Gram block associated with the fixed boundary words $(\boldsymbol{\alpha},\boldsymbol{\beta})$. Varying $\sigma$ and $\tau$ over all compatible color-preserving permutations reconstructs the entire tensor-product block. Different pairs of boundary words with the same occupation numbers give identical blocks, while varying the occupation numbers generates the different block types entering the full Gram matrix.\\

Figure~\ref{fig:gram_k2m2} illustrates this structure in the minimal case
$k=m=2$, for which
\begin{equation}
    G
    =
    G(2,d_1)
    \oplus
    [d_1d_2]^{\oplus 4}
    \oplus
    G(2,d_2) \ .
\end{equation}
Following the ordering introduced above, we first group the basis elements by their occupation vector $n$. For each fixed $n$, we then change the pair of boundary colored words $(\boldsymbol{\alpha},\boldsymbol{\beta})$ compatible with $n$. This gives, with the convention $ D^\sigma_{\boldsymbol{\alpha}, \boldsymbol{\beta}}\equiv D_{\boldsymbol A}$,
\begingroup
\newcommand{\cblue}{\cpdiagdot{cpdiag color 1}}
\newcommand{\cred}{\cpdiagdot{cpdiag color 2}}
\begin{equation}
\begin{aligned}
\big(
&D_{(\cblue,\cblue),(\cblue,\cblue)}^{e},
D_{(\cblue,\cblue),(\cblue,\cblue)}^{s},
D_{(\cred,\cblue),(\cred,\cblue)}^{e},
D_{(\cblue,\cred),(\cred,\cblue)}^{s},
\\
&
D_{(\cblue,\cred),(\cblue,\cred)}^{e},
D_{(\cred,\cblue),(\cblue,\cred)}^{s},
D_{(\cred,\cred),(\cred,\cred)}^{e},
D_{(\cred,\cred),(\cred,\cred)}^{s}
\big)\ ,
\end{aligned}
\end{equation}
\endgroup
corresponding to the order of the color-preserving permutations chosen in the lateral column and rows of Figure~\ref{fig:gram_k2m2}.
The monochromatic occupation patterns $n=(2,0)$ and $n=(0,2)$ give the ordinary Gram blocks $G(2,d_1)$ and $G(2,d_2)$, respectively. Within each of these blocks, the two basis elements correspond to the permutations $e,s\in S_2$: equal permutations produce two closed loops and hence the diagonal entries $d_i^2$, while different permutations produce a single loop and hence the off-diagonal entries $d_i$.
For the mixed occupation pattern $n=(1,1)$, instead, the permutation group associated with each color is $S_1$ and is therefore trivial. Each fixed
pair of boundary words thus gives the one-dimensional Gram block $G(1,d_1)\otimes G(1,d_2)=[d_1d_2]$.
Since there are $N_{(1,1)}=2$ possible boundary words, this block appears with multiplicity $N_{(1,1)}^2=4$, corresponding to the four central
$d_1d_2$ entries in Fig.~\ref{fig:gram_k2m2}. Matrix elements between different pairs of boundary words vanish.

\subsection{Weingarten Matrix}
Since the Weingarten matrix is defined as the Moore--Penrose pseudoinverse of the Gram matrix, the decomposition above immediately induces the corresponding block decomposition
\begin{align}
    \mathrm{Wg}
    & = \bigoplus_{\substack{n,\\ |n|=k}} 
    \left[
        \bigotimes_{i=1}^{m} \Wg(n_i,d_i)
    \right]^{\oplus N_{n}^2} \\
    \equiv &
    \bigoplus_{ n}
    \left(
    \mathbb{I}_{N_{ n}^2}
    \otimes
    \bigotimes_{i=1}^{m} \Wg(n_i,d_i)
    \right)\ , 
    \\ 
    &\mathrm{with }\quad N_{n}
    :=
    \frac{k!}{\prod_{i=1}^{m} n_i!} \ ,
\end{align}
where $\Wg(n_i,d_i)$
is the ordinary unitary Weingarten matrix on $n_i$ replicas in dimension $d_i$. 

In this way, when each subspace dimension $d_i$ is considered asymptotically large, one can use the standard result in Weingarten calculus
\begin{equation}
    \Wg_{\widehat{\sigma}_i,\widehat{\tau}_i}(n_i,d_i)
    =
    d_i^{-n_i-|\pi_i|}
    \left[
        \mathrm{Mob}(\pi_i)
        +
        O(d_i^{-2})
    \right],
\end{equation}
where
\begin{equation}
    \pi_i
    :=
    \widehat{\sigma}_i^{-1}\widehat{\tau}_i
    \in S_{n_i}\ ,
    \qquad
    |\pi_i|
    :=
    n_i-\#\operatorname{cycles}(\pi_i)\ .
    \label{eq:colour_relative_permutation}
\end{equation}
The function $\mathrm{Mob}$ is the Möbius function on permutations. If
$\pi\in S_r$ has cycle decomposition $\pi=c_1\cdots c_\ell$, then
\begin{equation}
    \mathrm{Mob}(\pi)
    :=
    \prod_{c\in \mathrm{Cycles}(\pi)}
    (-1)^{|c|-1}
    \mathrm{Cat}_{|c|-1} \ ,
    \label{eq:mobius_definition}
\end{equation}
where $\mathrm{Cat}_{q}:=\frac{1}{q+1}\binom{2q}{q}$
is the $q$-th Catalan number.
Therefore, for $\boldsymbol{A}=(\boldsymbol{\alpha},\boldsymbol{\beta};\sigma)$ and $\boldsymbol{B}=(\boldsymbol{\gamma},\boldsymbol{\delta};\tau)$,
\begin{equation}
\small
    \Wg_{\boldsymbol A,\boldsymbol B}
    =
    \delta_{\boldsymbol{\alpha},\boldsymbol{\gamma}}
    \delta_{\boldsymbol{\beta},\boldsymbol{\delta}}
    \prod_{i=1}^{m}
    d_i^{-n_i-|\pi_i|}
    \left[
        \mathrm{Mob}(\pi_i)
        +
        O(d_i^{-2})
    \right] \ .
    \label{eq:wg_coloured_asymptotic}
\end{equation}

\section{Random Quantum States with Symmetries}
\label{sec_quantumstates}

Random quantum states provide a natural setting for studying typical properties of high-dimensional quantum systems. In the presence of symmetries, however, the relevant notion of randomness is constrained by the decomposition of Hilbert space into symmetry sectors, leading to block-resolved ensembles.
In what follows, we study $k$-th moment operators of pure states and their application to bipartite entanglement.

\subsection{Pure state averages}

The unitary Weingarten calculus above gives a simple sector-resolved
calculus for random pure states. Let
\begin{equation}
    \ket{\psi_0}
    =
    \sum_{i=1}^m \sqrt{p_i}\,\ket{\eta_i},
    \quad
    p_i\geq 0,
    \quad
    \sum_{i=1}^m p_i=1,
    \label{eq:initial_state_sector_decomposition}
\end{equation}
where each $\ket{\eta_i}\in\HH_i$ is normalized whenever $p_i\neq 0$. 
The block-unitary pure-state ensemble is
\begin{equation}
    \ket{\psi(U)}
    :=
    U\ket{\psi_0}
    =
    \sum_{i=1}^m
    \sqrt{p_i}\,U_i\ket{\eta_i},
    \label{eq:block_haar_state_ensemble}
\end{equation}
with $U=\bigoplus_{i=1}^m U_i\in \mathcal{U}(\HH|\textbf{M})$ and the numbers $p_i$ are conserved by the ensemble. They are the sector weights of
the state.
For a triple $\boldsymbol A$, define
\begin{equation}
    n(\boldsymbol A)
    :=
    n(\boldsymbol{\alpha})
    =
    n(\boldsymbol{\beta}),
    \qquad
    p^{n(\boldsymbol A)}
    :=
    \prod_{i=1}^m p_i^{n_i(\boldsymbol A)},
\end{equation}
and introduce the rising factorial
\begin{equation}
    (d_i)_{r}
    :=
    d_i(d_i+1)\cdots(d_i+r-1),
    \qquad
    (d_i)_0:=1 .
    \label{eq:rising_factorial}
\end{equation}

Then, one can compute the average state and all its moments for the ensemble $\mathcal{U}(\HH|\textbf{M})$ as follows.

\begin{lemma}[Pure-state colored Weingarten calculus]
\label{lemma:pure_state_block_weingarten}
Let
\begin{equation}
    M_k(\boldsymbol p)
    :=
    \int_{\mathcal{U}(\HH|\textbf{M})}dU\,
    \left(\proj{\psi(U)}\right)^{\otimes k}
    \label{eq:pure_state_moment_operator}
\end{equation}
be the $k$-th moment operator of the block-unitary pure-state ensemble. Then
\begin{equation}
    M_k(\boldsymbol p)
    =
    \sum_{\boldsymbol A\in\mathcal{C}_k}
    \frac{
        p^{n(\boldsymbol A)}
    }{
        \prod_{i=1}^m (d_i)_{n_i(\boldsymbol A)}
    }
    D_{\boldsymbol A}.
    \label{eq:pure_state_block_weingarten_compact}
\end{equation}
\end{lemma}
Additionally, to recover the full unitary Haar-random moments, one can impose the following condition on the color weights.
\begin{corollary}[Averaging over sector weights]
\label{cor:Averaging_weights}
Let the sector weights be distributed according to the Diriclet distribution
\begin{equation}
    d\nu(\boldsymbol p)
    =
    \frac{1}{Z}
    \delta\left(\sum_{i=1}^m p_i-1\right)
    \prod_{i=1}^m p_i^{d_i-1}\,dp_i,
    \,
    Z=
    \frac{\prod_{i=1}^m\Gamma(d_i)}{\Gamma(d)},
\end{equation}
where $d=\sum_i d_i$.
Then, the average over sector weights restores the ordinary Haar pure-state moment on the full Hilbert space
\begin{equation}
    \E_{\boldsymbol p}[M_k(\boldsymbol p)]=
    \frac{P_{\sym}^{(k)}}{\binom{d+k-1}{k}},
\end{equation}
with $P_{\sym}^{(k)}=\frac{1}{k!}\sum_{\sigma\in S_k}\V_\sigma$.
\end{corollary}

\begin{figure}[t]
    \centering
    \includegraphics[width=0.9\linewidth]{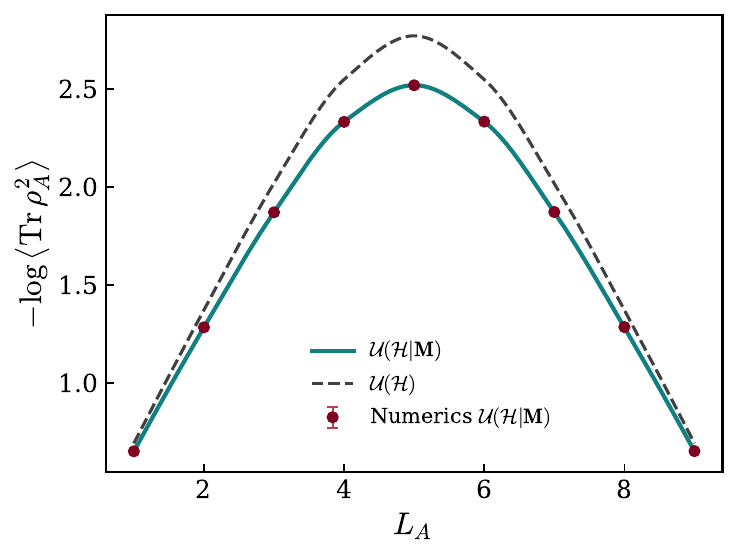}
    \caption{
Annealed second-Rényi Page curve $-\log \left\langle \Tr \rho_A^2 \right \rangle$ for a system of $L=10$ spin-$1/2$ degrees of freedom with fixed total charge $Q=4$.
The solid curve is obtained via Eq.~\eqref{eq:fixed_charge_average_purity}. The circles are numerical Haar samples in the fixed-charge sector. The gray dashed curve is the ordinary full-Haar Page curve on $\HH_A\otimes\HH_B$, shown for comparison.
}
\label{fig:fixed_charge_renyi_page}
\end{figure}

These results can be applied to several resource theories of quantum states~\cite{Chitambar_2019}, including entanglement~\cite{BianchiDona2019,BianchiDonaKumar2024,Russotto_2025,bianchi2026typicalentanglemententropycharge, MurcianoCalabresePiroli2022,LanglettRodriguezNieva2025,Hejazi_2022,Iizuka_2026}, magic~\cite{iannotti2026nonstabilizernessu1symmetrychaotic,Esposito_2024,Cepollaro_2024,Leone_2021,tirrito2025universalspreadingnonstabilizernessquantum,cepollaro2025stabilizerentropysubspaces,Cusumano2026Probes}, coherence~\cite{Andreadakis_2023,aditya2026coherencedynamicsquantummanybody}, non-Gaussianity~\cite{ares2026nongaussianityrandomquantumstates}, and whenever the corresponding resource quantifier can be expressed in terms of finite moments of the state.

\subsection{Entanglement}
\label{subsec:entanglement}

We now apply the block-unitary formalism to the second-Rényi entropy of a
bipartite system in a fixed-charge sector. 
Let
\begin{equation}
    \HH=\HH_A\otimes\HH_B,
    \qquad
    Q=Q_A+Q_B ,
\end{equation}
and denote by $\HH_q$ the eigenspace of total charge $q$. The charge constraint
induces the decomposition
\begin{equation}
    \HH_q
    =
    \bigoplus_{a\in\mathcal I_q}
    \HH_A^{(a)}\otimes \HH_B^{(q-a)} ,
    \label{eq:fixed_charge_bipartite_decomp}
\end{equation}
where $a$ labels the charge in subsystem $A$, and $\mathcal I_q$ is the set of
allowed charge splittings. We write
\begin{equation}
    d_A(a):=\dim \HH_A^{(a)},\quad
    d_B(q-a):=\dim \HH_B^{(q-a)},
\end{equation}
\begin{equation}
    D_a:=d_A(a)d_B(q-a),
    \quad
    d_q:=\dim\HH_q=\sum_{a\in\mathcal I_q}D_a .
\end{equation}
For a Haar-random pure state $\psi_q=\ket{\psi_q}\bra{\psi_q}$ in $\HH_q$, the reduced state on $A$ is $\rho_A=\Tr_B \psi_q$ and its second-Rényi entropy is $S_2(A)=-\log \Tr \rho_A^2 $.
The fixed-charge Haar average~\cite{Liu_2020} gives
\begin{equation}
\small
    \E_q[\Tr \rho_A^2]
    =
    \frac{
    \sum_{a\in\mathcal I_q}
    d_A(a)d_B(q-a)
    \left[d_A(a)+d_B(q-a)\right]
    }{
    d_q(d_q+1)
    }
    \label{eq:fixed_charge_average_purity},
\end{equation}
with derivation in App.~\ref{app:fixed_charge_renyi2}.
This is the second-Rényi version of the fixed-charge Page computation of Bianchi--Donà~\cite{BianchiDona2019}. In their formulation, a Haar-random state in $\HH_q$ can be written as
\begin{equation}
    \ket{\psi_q}
    =
    \sum_{a\in\mathcal I_q}
    \sqrt{p_a}\,\ket{\phi_a},
    \qquad
    \ket{\phi_a}\in
    \HH_A^{(a)}\otimes\HH_B^{(q-a)} ,
    \label{eq:bianchi_dona_decomposition}
\end{equation}
where the $\ket{\phi_a}$ are independently Haar-distributed in the blocks and the weights $\boldsymbol p=(p_a)$ follow the Dirichlet distribution with parameters $D_a$. In the notation of this paper, the label $a$ is a color. For spin-$1/2$ systems of size $L=L_A+L_B$ with fixed magnetization, one may take
\begin{equation}
    d_A(a)=\binom{L_A}{\frac{L_A+a}{2}},
    \quad
    d_B(q-a)=\binom{L_B}{\frac{L_B+q-a}{2}},
\end{equation}
with the convention that the binomial coefficient is zero when its lower
argument is not an integer in the allowed range (see Fig.~\ref{fig:fixed_charge_renyi_page} for an example).

\section{Block-Haar Randomness in ETH}
\label{sec_ETH}

A conceptually different application of the framework developed here arises in the context of the Eigenstate Thermalization Hypothesis (ETH). ETH concerns the matrix elements of physical observables in the energy eigenbasis~\cite{Srednicki_1999},
\begin{equation}
    O_{ij} = \langle E_i | O | E_j \rangle, \qquad H |E_i\rangle = E_i |E_i\rangle ,
\end{equation}
whose pseudorandom behavior in chaotic many-body systems underlies the emergence of thermal expectation values. 
The statistical interpretation underlying ETH has a long history, connecting the foundations of pure-state statistical mechanics with random-matrix
approaches to quantum chaos~\cite{DAlessio_2016}. Rather than resolving the full microscopic dynamics, one may consider a fictitious ensemble of slightly perturbed Hamiltonians around a given system; averaging over such an ensemble provides a simpler statistical description while retaining the universal features of chaotic dynamics. In random matrix theory, these ensembles are taken to be invariant under changes of basis. Here we focus on the unitary symmetry class, appropriate in the absence of time-reversal symmetry~\footnote{In the presence of time-reversal
symmetry, the corresponding invariant ensemble may instead belong to the orthogonal or symplectic symmetry class.}. This invariance leads to eigenstates uniformly distributed on the unit sphere. \\
For realistic many-body Hamiltonians, however, such global randomness is too strong: eigenvectors cannot be regarded as uniformly distributed over the full Hilbert space, as their statistical properties retain a nontrivial dependence on energy. Randomness is instead constrained to small energy windows, within which nearby eigenstates can be effectively regarded as randomly mixed.

The randomness of the Hamiltonian eigenvectors is reflected in the statistical properties of the matrix elements of observables in the energy basis. In particular, an observable $O$ inherits a rotational invariance, for which its statistical properties are unchanged under the transformation $O\mapsto U^\dagger O U$. For realistic many-body systems this invariance should hold only locally in energy: the corresponding transformation is therefore not a global Haar rotation, but a block-unitary transformation.
This property, referred to as \emph{local rotational invariance} \cite{ vallini2025refinementseigenstatethermalizationhypothesis}, can be expressed as
\begin{align}
    O \sim O^U &= U_{\mathrm{loc}}^\dagger O  U_{\mathrm{loc}}  \\
\mathcal{H} = \bigoplus_{I=1}^{m} \mathcal{H}_I, \quad &\qquad U_{\mathrm{loc}} = \bigoplus_{I=1}^{m} U_I ,
    \label{eq:local_rotational_invariance}
\end{align}
where $U_I$ are independent Haar-random unitaries acting on subspaces $\mathcal{H}_I$ spanned by eigenstates belonging to the same mesoscopic energy window. In this way, $U_{\mathrm{loc}}$ defines an energy-resolved random ensemble whose blocks correspond to local regions of the spectrum, as depicted in Fig.~\ref{fig_ETH}. 
It is important to stress that $U_{\mathrm{loc}}$ does not represent the physical time evolution $e^{-iHt}$. Rather, it is a random basis transformation implementing the local-in-energy mixing entering the statistical description of observables. Its physical origin is therefore conceptually different from the symmetry-constrained transformations considered in the previous sections. Mathematically, however, the two problems share precisely the same block-Haar structure.

Local rotational invariance can be used for evaluating the multipoint matrix-element correlations entering the full ETH ansatz~\cite{FoiniKurchan2019}. Indeed, the $k$-point moments of the randomly rotated observable can be written directly as matrix elements of the block-Haar twirling map,
\begin{align}
\label{eq:ETH_twirl}
    &\overline{O^U_{i_1 j_1}\cdots O^U_{i_k j_k}} = \\  &\langle E_{i_1}\otimes\cdots\otimes E_{i_k}| \Phi_{\mathcal{U}(\mathcal{H}|\mathbf{M})}^{ (k)}\left(O^{\otimes k}\right) |E_{j_1}\otimes\cdots\otimes E_{j_k}\rangle ,
    \notag
\end{align}
where the overline denotes the average over the block-unitaries ensemble.
In Ref.~\cite{vallini2025refinementseigenstatethermalizationhypothesis}, these averages were evaluated by explicitly resolving the contributions associated with different energy blocks. The block-Weingarten calculus developed here places these calculations within a systematic arbitrary-order framework: once the energy decomposition is specified, the corresponding moments follow directly from the symmetry-resolved Schur--Weyl structure
derived in Eq~\eqref{eq:WC}.

\begin{figure}
    \centering
    \includegraphics[width=0.8\linewidth]{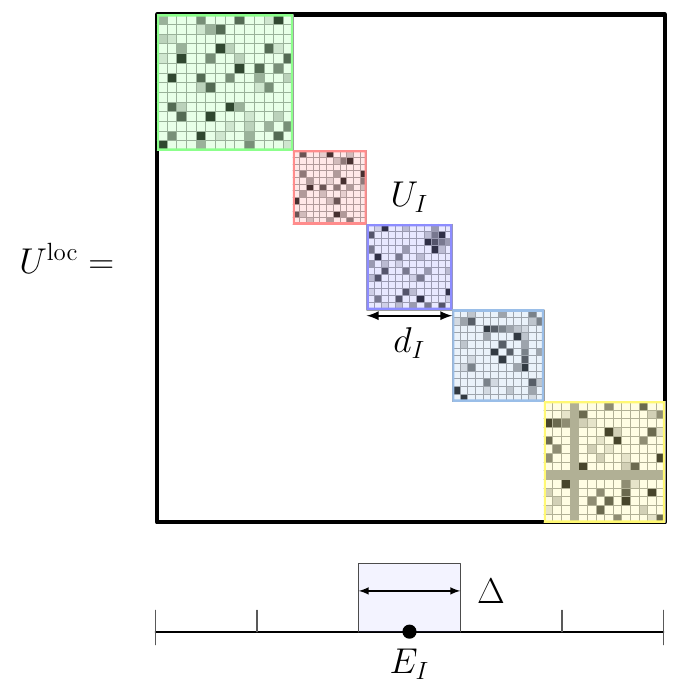}
    \caption{Block-Haar structure of $U_{\mathrm{loc}}$ in Eq.~\eqref{eq:local_rotational_invariance}. Each block $U_I$ acts on the energy eigenstates within a window of width $\Delta$ centered at $E_I$, and is independently Haar distributed over $\mathcal U(d_I)$. The lower axis shows the energy window associated with one of the blocks above.}
    \label{fig_ETH}
\end{figure}
To make contact with the notation of the previous sections, the colors now label mesoscopic energy windows rather than symmetry sectors. We can divide the energy spectrum in $m$ intervals of length $\Delta$ and identify each interval with a mesoscopic energy value $E_I, I=1\dots m$ which can be taken to be the middle value, as illustrated in Fig.~\ref{fig_ETH}. Each eigenstate is associated to a color specifying the energy window to which it belongs. We introduce the corresponding projectors
\begin{equation}
    \Pi_I = \sum_{i:\, E_i \in E^{\Delta}_I} |E_i\rangle\langle E_i|, \qquad d_I = \Tr \Pi_I .
    \label{eq:ETH_energy_projectors}
\end{equation}
with $E^{\Delta}_I = [ E_I - \Delta/2, E_I + \Delta/2)$. The block dimension $d_I$ counts the number of eigenstates contained in the corresponding energy windows.
The color constraints appearing in the block-Weingarten calculus therefore translate into constraints on the mesoscopic energies carried by the indices of the ETH correlation functions.

Evaluating Eq.~\eqref{eq:ETH_twirl} through the Weingarten calculus developed in Eq.~\eqref{eq:WC}, at leading order in the block dimensions, one directly obtains the correlation structures associated with local rotational invariance. For the global Haar ensemble~\cite{DAlessio_2016}, the same calculation gives energy-independent correlations, while the block-unitary ensemble naturally reproduces the energy dependence of the ETH ansatz. For instance, for $k=1,2$:
\begin{subequations}
\label{eq:ETH_ansatz}
\begin{align} \label{eq_ETH1}
    \overline{O_{ii}} &= K_1(E_I)\\ \label{eq_ETH2}
    \overline{O_{ij}O_{ji}} &= K^2_1(E_I) \delta_{ij} + K_2(E_I,E_J)
\end{align}
with $E_{i} \in E^\Delta_{I}$, $E_{j} \in E^\Delta_{J}$ and
\begin{align}
    K_1(E_I) &= \text{Tr}(\hat \Pi_I O)\\
    K_2(E_I,E_J) &= \text{Tr}(\hat \Pi_I O\hat \Pi_J O) -\frac{\delta_{IJ}}{d_I}K^2_1(E_I)
\end{align}
\end{subequations}
having defined $\hat \Pi_I = \Pi_I/d_I$. 
We notice that $\Pi_I O$ is the block of $O$ containing all the states associated to the mesoscopic energy $E_I$, so that traces are local on the block, for instance  $\text{Tr}( \Pi_I O) = \sum_{k: E_k \in E^\Delta_I} O_{kk}$. The quantities $K_k$ were introduced in
\cite{vallini2025refinementseigenstatethermalizationhypothesis} and referred to
as \textit{local free cumulants}, based on their connection with free probability, which provides a natural framework for describing correlations
between noncommuting random variables~\cite{NicaSpeicher2006}. 
We illustrate how to obtain Eqs.~\eqref{eq:ETH_ansatz} in App.~\ref{app:ETH}.

\section{Conclusion and outlook}
\label{sec_conclusion}
In this work, we developed a Weingarten calculus for Haar-random unitaries constrained by a decomposition of the Hilbert space. The central result is a symmetry-resolved analogue of Schur--Weyl duality for the Block-Unitary Ensemble. Rather than the ordinary permutation algebra associated with the full unitary group, the commutant is generated by color-preserving permutation diagrams between pairs of colored words, where each color labels one of the Hilbert-space sectors. This provides a constructive description of the invariant structures that survive the block-Haar average and turns arbitrary moments of Abelian symmetry-constrained random unitaries into a combinatorial problem.
The colored structure leads to a particularly simple organization of the corresponding Weingarten calculus. The Hilbert--Schmidt Gram matrix vanishes between diagrams with different boundary colored words, while within each fixed pair of boundaries it factorizes into a tensor product of ordinary unitary Gram matrices. 

We illustrated this formalism first for random quantum states subject to symmetry constraints.
We then show that the same mathematical structure arises from a conceptually different constraint in the local rotational invariance formulation of the Eigenstate Thermalization Hypothesis.
More broadly, the Block-Unitary Ensemble provides a natural reference ensemble whenever randomness is compatible with a decomposition of Hilbert space rather than with the full unitary group. 
A closely related equal-block structure was previously considered in the study of random unitary quantum channels, where independent Haar-distributed unitaries were organized into block random matrices and treated through an extension of the graphical Weingarten calculus~\cite{COLLINS_2013}. More recently, the same equal-dimensional block structure has appeared in quantum cryptography in the form of the unitary Haar cipher~\cite{foxman2026quantumlazysamplingpath}.

This work may provide a broader calculus for structured randomness in quantum many-body systems, in which the constraints imposed by symmetries, energy resolution, or restricted gate sets are incorporated directly into the random ensemble rather than treated as corrections to fully Haar-random dynamics.

Several directions follow from this perspective. \\
This framework can be systematically applied to fluctuations of symmetry-resolved entanglement~\cite{BianchiDona2019,BianchiDonaKumar2024}, symmetry-constrained scrambling and out-of-time-order correlators~\cite{Gisti_2025, KhemaniVishwanathHuse2018, RakovszkyPollmannKeyserlingk2018}, symmetry-preserving unitary designs~\cite{Hearth_2025,Mitsuhashi_2025,Marvian2022}, and randomized measurement protocols performed in the presence of conservation laws~\cite{Low2022}. \\
On the thermalization side, recent developments have connected higher-order ETH correlations with free probability~\cite{NicaSpeicher2006}, with free cumulants providing a natural organization of the smooth multipoint correlations entering full ETH~\cite{Pappalardi_2022,FritzschProsenPappalardi2025,AlvesFritzschClaeys2025,JindalHosur2024,ValliniPappalardi2026,Fritzsch2026}.
The leading-order symmetry-resolved Weingarten calculus could provide a systematic connection between higher-order ETH correlations and operator-valued free probability, in analogy with the way ordinary Weingarten calculus connects the global Haar-random setting to scalar-valued free probability. A closely related connection has been established for the same subgroup of the unitary group in Eq.~\eqref{eq:blockunitary} with a different representation, obtained by replacing the direct-sum block structure in Eq.~\eqref{eq:embeddingsum} with a tensor product of local unitaries~\cite{nechita2025tensorfreeprobabilitytheory}. A second open direction is to formulate a smooth version of the Block-Unitary Ensemble directly at the level of its commutant structure, replacing sharp energy windows by a continuously energy-resolved description~\cite{vallini2025refinementseigenstatethermalizationhypothesis}. Finally, it would be interesting to extend the present framework to ETH in the presence of symmetries. While Abelian conserved quantities can be incorporated by working within fixed symmetry sectors~\cite{LeBlondRigol2020,Onoda2026,FukushimaHamazaki2023,Gomez2026}, non-Abelian symmetries require a more refined treatment due to the presence of noncommuting conserved charges~\cite{Murthy2023, PatilRigol2025,Lasek_2026,JindalHosur2026}.\\
It would also be interesting to determine the corresponding diagrammatic and moment compact structures for other block ensembles, including orthogonal and symplectic Haar ensembles~\cite{CollinsSniady2006} as well as symmetry-preserving Clifford and random Gaussian unitaries ensembles~\cite{Bittel_2026,Gross_2021,MitsuhashiYoshioka2023,walter2025randompurificationchannelarbitrary, sierant2026theorymatchgatecommutant,braccia2026commutantfermionicgaussianunitaries,lastres2026geometryfreefermioncommutants}.

\section{Acknowledgments}
The authors are truly grateful to Silvia Pappalardi for insightful discussions and helpful suggestions throughout the development of this work.
They also thank Alioscia Hamma, Gianluca Esposito, Xhek Turkeshi, Beatrice Magni, Angelo Russotto and Arnau Lira Solanilla for their valuable feedbacks.
D.I. acknowledges the support of Research Grants in Germany, 2026 (57812126) from the German Academic Exchange Service (DAAD). E.V. is supported by the Deutsche Forschungsgemeinschaft (DFG, German Research Foundation) under Germany’s Excellence Strategy - Cluster of Excellence Matter and Light for Quantum Computing (ML4Q) EXC 2004/1 -390534769, and DFG Collaborative Research Center (CRC) 183 Project No. 277101999 - project B02.\\
AI was used during the preparation of this manuscript to improve clarity and presentation, and help check mathematical consistency and intermediate derivations. AI was not used to generate the scientific ideas and results.

\bibliography{reference}
\bibliographystyle{apsrev4-1}

\appendix

\section{Proof of Theorem \ref{thm:block_commutant}}\label{app:additional-proofs}
\begin{proof}
Since the projectors $\Pi_i$ are mutually orthogonal and satisfy
$\sum_{i=1}^m \Pi_i=\id_{\HH}$, the multifactor projectors
\begin{equation}
    \Pi_{\boldsymbol{\alpha}}
    =
    \Pi_{\alpha_1}\otimes\cdots\otimes\Pi_{\alpha_k},
    \qquad
    \boldsymbol{\alpha}\in\{1,\dots,m\}^k ,
\end{equation}
are mutually orthogonal and resolve the identity on $\HH^{\otimes k}$:
\begin{equation}
    \Pi_{\boldsymbol{\alpha}}\Pi_{\boldsymbol{\beta}}
    =
    \delta_{\boldsymbol{\alpha},\boldsymbol{\beta}}
    \Pi_{\boldsymbol{\alpha}},
    \quad
    \sum_{\boldsymbol{\alpha}\in\{1,\dots,m\}^k}
    \Pi_{\boldsymbol{\alpha}}
    =
    \id_{\HH^{\otimes k}} .
\end{equation}
Thus every $X\in\BB(\HH^{\otimes k})$ decomposes into colored-word blocks
\begin{equation}
    X=\sum_{\boldsymbol{\alpha},\boldsymbol{\beta}}
    X_{\boldsymbol{\alpha}\boldsymbol{\beta}},
    \quad X_{\boldsymbol{\alpha}\boldsymbol{\beta}} :=\Pi_{\boldsymbol{\alpha}}X\Pi_{\boldsymbol{\beta}} \in \BB(\HH_{\boldsymbol{\beta}},\HH_{\boldsymbol{\alpha}}).
\end{equation}
Let $U=\bigoplus_{i=1}^m U_i\in \mathcal{U}(\HH|\textbf{M})$. Since $U$ is block
diagonal, every $\Pi_i$ commutes with $U$, and hence every
$\Pi_{\boldsymbol{\alpha}}$ commutes with $U^{\otimes k}$. Moreover, the
restriction of $U^{\otimes k}$ to $\HH_{\boldsymbol{\alpha}}$ is
\begin{equation}
    U_{\boldsymbol{\alpha}}
    :=
    U_{\alpha_1}\otimes\cdots\otimes U_{\alpha_k}.
\end{equation}
Therefore $X\in\Comm_k(\mathcal{U}(\HH|\textbf{M}))$ if and only if each block satisfies
\begin{equation}
    U_{\boldsymbol{\alpha}}
    X_{\boldsymbol{\alpha}\boldsymbol{\beta}}
    =
    X_{\boldsymbol{\alpha}\boldsymbol{\beta}}
    U_{\boldsymbol{\beta}}
    \qquad
    \forall U\in \mathcal{U}(\HH|\textbf{M}) .
    \label{eq:app_block_intertwiner_condition}
\end{equation}

We first show that no block can connect two words with different occupation
vectors. For $\theta=(\theta_1,\dots,\theta_m)\in\mathbb{R}^m$, consider the
central block unitary
\begin{equation}
    Z(\theta)
    :=
    \bigoplus_{i=1}^m e^{i\theta_i}\id_{\HH_i}
    \in \mathcal{U}(\HH|\textbf{M}) .
\end{equation}
On $\HH_{\boldsymbol{\alpha}}$, it acts as the scalar
\begin{equation}
    Z(\theta)_{\boldsymbol{\alpha}}
    =
    e^{i\sum_i n_i(\boldsymbol{\alpha})\theta_i}
    \id_{\HH_{\boldsymbol{\alpha}}}.
\end{equation}
Substituting $Z(\theta)$ into
Eq.~\eqref{eq:app_block_intertwiner_condition} gives
\begin{equation}
    e^{i\sum_i n_i(\boldsymbol{\alpha})\theta_i}
    X_{\boldsymbol{\alpha}\boldsymbol{\beta}}
    =
    e^{i\sum_i n_i(\boldsymbol{\beta})\theta_i}
    X_{\boldsymbol{\alpha}\boldsymbol{\beta}}
    \qquad
    \forall \theta\in\mathbb{R}^m .
\end{equation}
If $n(\boldsymbol{\alpha})\neq n(\boldsymbol{\beta})$, the two characters are
different, and the above identity forces
$X_{\boldsymbol{\alpha}\boldsymbol{\beta}}=0$. Hence every element of the
commutant is block diagonal with respect to the occupation decomposition
\begin{equation}
    \HH^{\otimes k}
    =
    \bigoplus_{\substack{n\in\mathbb{N}^m\\ |n|=k}}
    \HH_n,
    \qquad
    \HH_n
    :=
    \bigoplus_{\boldsymbol{\alpha}:\,n(\boldsymbol{\alpha})=n}
    \HH_{\boldsymbol{\alpha}} .
\end{equation}
It remains to identify the blocks with fixed occupation vector. Fix $\boldsymbol{\alpha},\boldsymbol{\beta}\in\{1,\dots,m\}^k$ such that $n(\boldsymbol{\alpha})=n(\boldsymbol{\beta})=n$. For each color $i$, write
\begin{align}
    P_i(\boldsymbol{\alpha})
    &=
    \{p_{i,1}<\cdots<p_{i,n_i}\},\\
    P_i(\boldsymbol{\beta})
    &=
    \{q_{i,1}<\cdots<q_{i,n_i}\},
\end{align}
where $P_i(\boldsymbol{\alpha})$ is the set of positions occupied by color $i$ in $\boldsymbol{\alpha}$, and similarly for $\boldsymbol{\beta}$.
Equivalently, we can write it as
\begin{equation}
    P_i(\boldsymbol{\alpha})
    :=
    \{r\in\{1,\ldots,k\}:\alpha_r=i\}.
\end{equation}
Define
\begin{equation}
    K_n
    :=
    \bigotimes_{i=1}^m \HH_i^{\otimes n_i},
\end{equation}
with the convention $\HH_i^{\otimes 0}=\C$. Let $T_{\boldsymbol{\alpha}}:\HH_{\boldsymbol{\alpha}}\to K_n$ be the unitary
which groups tensor factors according to their color:
\begin{equation}
    T_{\boldsymbol{\alpha}}
    \ket{x_1\otimes\cdots\otimes x_k}
    =
    \bigotimes_{i=1}^m
    \ket{x_{p_{i,1}}\otimes\cdots\otimes x_{p_{i,n_i}}}.
\end{equation}
Then $ T_{\boldsymbol{\alpha}} U_{\boldsymbol{\alpha}}T_{\boldsymbol{\alpha}}^{-1}= \bigotimes_{i=1}^m U_i^{\otimes n_i}$
and the same relation holds with $\boldsymbol{\beta}$.
Let $A\in\BB(\HH_{\boldsymbol{\beta}},\HH_{\boldsymbol{\alpha}})$ satisfy
Eq.~\eqref{eq:app_block_intertwiner_condition}, and set $Y:= T_{\boldsymbol{\alpha}} A T_{\boldsymbol{\beta}}^{-1} \in \BB(K_n)$, then
\begin{equation}
    \left(\bigotimes_{i=1}^m U_i^{\otimes n_i}\right)Y
    =
    Y\left(\bigotimes_{i=1}^m U_i^{\otimes n_i}\right),
\end{equation}
$\forall (U_1,\dots,U_m)\in \mathcal{U}(\HH|\textbf{M}) $.
Since the factors $U_i$ vary independently, the commutant on $K_n$ factorizes:
\begin{equation}
\begin{aligned}
    &
    \left\{
        Y\in\BB(K_n):
        \left[
            Y,\bigotimes_{i=1}^m U_i^{\otimes n_i}
        \right]=0
        \ \forall (U_1,\dots,U_m)
    \right\}
    \\
    &\qquad =
    \bigotimes_{i=1}^m
    \left\{
        Y_i\in\BB(\HH_i^{\otimes n_i}):
        [Y_i,U_i^{\otimes n_i}]=0
        \ \forall U_i\in\mathcal{U}(\HH_i)
    \right\}.
\end{aligned}
\end{equation}
By the ordinary finite-dimensional Schur--Weyl commutant theorem,
\begin{equation}
\begin{aligned}
     \{Y_i\in\BB(\HH_i^{\otimes n_i})&:[Y_i,U_i^{\otimes n_i}]=0 \ \forall U_i\in\mathcal{U}(\HH_i)
    \}\\
    &=
    \Span\{\V_{\tau_i}^{(i)}:\tau_i\in S_{n_i}\},
\end{aligned}
\end{equation}
where $\V_{\tau_i}^{(i)}$ permutes the $n_i$ copies of $\HH_i$. Therefore
\begin{equation}
    A
    \in
    \Span\left\{
        T_{\boldsymbol{\alpha}}^{-1}
        \left(
            \bigotimes_{i=1}^m \V_{\tau_i}^{(i)}
        \right)
        T_{\boldsymbol{\beta}}
        :
        \tau_i\in S_{n_i}
    \right\}.
\end{equation}
We now identify these reordered permutations with projected global permutations. Given $\tau_i\in S_{n_i}$ for each $i$, define
$\sigma\in S_k$ by
\begin{equation}
    \sigma(q_{i,b})
    =
    p_{i,\tau_i(b)}
    \quad
    \forall i=1,\dots,m,\quad b=1,\dots,n_i .
\end{equation}
Then $\sigma$ maps the positions of color $i$ in $\boldsymbol{\beta}$ to the positions of color $i$ in $\boldsymbol{\alpha}$. Equivalently,
\begin{equation}
    \alpha_r=\beta_{\sigma^{-1}(r)}
    \qquad
    \forall r=1,\dots,k .
    \label{eq:app_colour_preserving_condition}
\end{equation}
Using the convention
\begin{equation}
    \V_\sigma
    \ket{x_1,\dots,x_k}
    =
    \ket{x_{\sigma^{-1}(1)},\dots,x_{\sigma^{-1}(k)}},
\end{equation}
one checks on elementary tensors that
\begin{equation}
    T_{\boldsymbol{\alpha}}^{-1}
    \left(
        \bigotimes_{i=1}^m \V_{\tau_i}^{(i)}
    \right)
    T_{\boldsymbol{\beta}}
    =
    \Pi_{\boldsymbol{\alpha}}\V_\sigma\Pi_{\boldsymbol{\beta}} .
\end{equation}
Indeed, if $r=p_{i,a}$, then
$\sigma^{-1}(r)=q_{i,\tau_i^{-1}(a)}$, so both sides place the tensor $x_{q_{i,\tau_i^{-1}(a)}}$ in the output position $p_{i,a}$. Conversely, any $\sigma\in S_k$ satisfying
Eq.~\eqref{eq:app_colour_preserving_condition} maps the positions of each color in $\boldsymbol{\beta}$ bijectively onto the positions of the same color in $\boldsymbol{\alpha}$. It therefore determines a unique tuple $(\tau_1,\dots,\tau_m)$ with $\tau_i\in S_{n_i}$. Hence, for fixed $\boldsymbol{\alpha},\boldsymbol{\beta}$ with
$n(\boldsymbol{\alpha})=n(\boldsymbol{\beta})$, the corresponding intertwiner space is spanned by
\begin{equation}
    \Pi_{\boldsymbol{\alpha}}\V_\sigma\Pi_{\boldsymbol{\beta}},
    \qquad
    \alpha_r=\beta_{\sigma^{-1}(r)}
    \quad \forall r .
\end{equation}
Summing over all pairs of colored words gives
\begin{equation}
    \Comm_k(\mathcal{U}(\HH|\textbf{M}))
    \subseteq \mathcal{D}_k
\end{equation}
where $\mathcal{D}_k=\Span\left\{
       \Pi_{\boldsymbol{\alpha}}\V_\sigma\Pi_{\boldsymbol{\beta}}:
        \sigma\in S_k,\,
        \alpha_r=\beta_{\sigma^{-1}(r)}
        \ \forall r
    \right\} $ is the linear span of the color-preserving permutation diagrams.
We prove the reverse inclusion. For every $U\in \mathcal{U}(\HH|\textbf{M})$, the projectors $\Pi_{\boldsymbol{\alpha}}$ and $\Pi_{\boldsymbol{\beta}}$ commute with $U^{\otimes k}$. Moreover,
\begin{equation}
    \V_\sigma U^{\otimes k}
    =
    U^{\otimes k}\V_\sigma,
\end{equation}
because the same operator $U$ acts on each tensor factor. Therefore
\begin{equation}
\begin{aligned}
    \Pi_{\boldsymbol{\alpha}}\V_\sigma
    \Pi_{\boldsymbol{\beta}}U^{\otimes k}
    &=
    \Pi_{\boldsymbol{\alpha}}\V_\sigma
    U^{\otimes k}\Pi_{\boldsymbol{\beta}}
    =
    \Pi_{\boldsymbol{\alpha}}U^{\otimes k}
    \V_\sigma\Pi_{\boldsymbol{\beta}}  \\
    &=
    U^{\otimes k}
    \Pi_{\boldsymbol{\alpha}}\V_\sigma\Pi_{\boldsymbol{\beta}}.
\end{aligned}
\end{equation}
Thus every color-preserving projected permutation belongs to
$\Comm_k(\mathcal{U}(\HH|\textbf{M}))$.
It remains only to rewrite the result in occupation-sector form. For fixed $\boldsymbol{\beta}$ and $\sigma$, the permutation $\V_\sigma$ maps
$\HH_{\boldsymbol{\beta}}$ onto $\HH_{\boldsymbol{\gamma}}$, where $\gamma_r:=\beta_{\sigma^{-1}(r)}$.
Whenever this holds, $n(\boldsymbol{\alpha})=n(\boldsymbol{\beta})$, since
$\sigma$ only permutes tensor positions. Conversely, if
$n(\boldsymbol{\alpha})=n(\boldsymbol{\beta})$, then the nonzero operators
$\Pi_{\boldsymbol{\alpha}}\V_\sigma\Pi_{\boldsymbol{\beta}}$ are precisely the permutations that map the word $\boldsymbol{\beta}$ to the word
$\boldsymbol{\alpha}$. Since all blocks with different occupation vector
vanish, the commutant decomposes as
\begin{equation}
\small
    \Comm_k(\mathcal{U}(\HH|\textbf{M}))
    =
    \bigoplus_{\substack{n\in\mathbb{N}^m\\ |n|=k}}
    \Span
    \left\{
        \Pi_{\boldsymbol{\alpha}}\V_\sigma\Pi_{\boldsymbol{\beta}}:
        n(\boldsymbol{\alpha})=n(\boldsymbol{\beta})=n
    \right\}.
\end{equation}
Terms in the last span that are not color preserving are zero, so including
them does not change the span.
\end{proof}

\section{Proof of Proposition \ref{prop:block_commutant_gram}}
\label{app:Gramm_sum}

\begin{proof}
We defined for a colored word $\boldsymbol{\alpha}$ and each color $i=1,\ldots,m$ the set of positions occupied by color $i$ as
\begin{equation}
    P_i(\boldsymbol{\alpha})
    :=
    \{r\in\{1,\ldots,k\}:\alpha_r=i\}.
\end{equation}
We first show explicitly how the color-preserving condition induces the restricted bijections in Eq.~\eqref{eq:colour_restricted_sigma}, that is
\begin{align}
    \sigma_i
    &:=
    \left.\sigma\right|_{P_i(\boldsymbol{\beta})}
    :
    P_i(\boldsymbol{\beta})
    \longrightarrow
    P_i(\boldsymbol{\alpha}), \\ 
    \tau_i
    &:=
    \left.\tau\right|_{P_i(\boldsymbol{\delta})}
    :
    P_i(\boldsymbol{\delta})
    \longrightarrow
    P_i(\boldsymbol{\gamma}).
\end{align}
Let $q\in P_i(\boldsymbol{\beta})$, so that $\beta_q=i$. Setting $r=\sigma(q)$ and using $\alpha_r=\beta_{\sigma^{-1}(r)}$, we obtain
\begin{equation}
    \alpha_{\sigma(q)}
    =
    \beta_q
    =
    i.
\end{equation}
Hence $\sigma(q)\in P_i(\boldsymbol{\alpha})$, and therefore $\sigma$ maps $P_i(\boldsymbol{\beta})$ into $P_i(\boldsymbol{\alpha})$. Conversely, if $r\in P_i(\boldsymbol{\alpha})$, then
\begin{equation}
    \beta_{\sigma^{-1}(r)}
    =
    \alpha_r
    =
    i,
\end{equation}
so that $\sigma^{-1}(r)\in P_i(\boldsymbol{\beta})$. Thus $\sigma_i$ is a bijection between the positions carrying color $i$ on the two boundary words. The same argument applies to $\tau_i$.

We now evaluate the Gram element. \\
From $\Pi_{\boldsymbol{\alpha}}^\dagger=\Pi_{\boldsymbol{\alpha}}$ and $\V_\sigma^\dagger=\V_{\sigma^{-1}}$, we have
$
    D_{\boldsymbol{A}}^\dagger
    =
    \Pi_{\boldsymbol{\beta}}
    \V_{\sigma^{-1}}
    \Pi_{\boldsymbol{\alpha}},
$
and hence
\begin{equation}
    G_{\boldsymbol{A},\boldsymbol{B}}
    =
    \Tr\!\left(
        \Pi_{\boldsymbol{\beta}}
        \V_{\sigma^{-1}}
        \Pi_{\boldsymbol{\alpha}}
        \Pi_{\boldsymbol{\gamma}}
        \V_\tau
        \Pi_{\boldsymbol{\delta}}
    \right).
    \label{eq:gram_proof_first_step}
\end{equation}
Since projectors corresponding to distinct colored words are orthogonal, i.e. $\Pi_{\boldsymbol{\alpha}}
    \Pi_{\boldsymbol{\gamma}} =\delta_{\boldsymbol{\alpha},\boldsymbol{\gamma}}
    \Pi_{\boldsymbol{\alpha}}$, 
the Gram element vanishes unless $\boldsymbol{\alpha}=\boldsymbol{\gamma}$. We may therefore write
\begin{equation}
    G_{\boldsymbol{A},\boldsymbol{B}}
    =
    \delta_{\boldsymbol{\alpha},\boldsymbol{\gamma}}
    \Tr\!\left(
        \Pi_{\boldsymbol{\beta}}
        \V_{\sigma^{-1}}
        \Pi_{\boldsymbol{\alpha}}
        \V_\tau
        \Pi_{\boldsymbol{\delta}}
    \right).
\end{equation}

For a color-preserving diagram, the permutation operator maps the tensor-product sector associated with the lower boundary word onto that associated with the upper boundary word. In particular,
\begin{equation}
    \V_{\sigma^{-1}}\Pi_{\boldsymbol{\alpha}}
    =
    \Pi_{\boldsymbol{\beta}}\V_{\sigma^{-1}},
    \qquad
    \Pi_{\boldsymbol{\alpha}}\V_\tau
    =
    \V_\tau\Pi_{\boldsymbol{\delta}}.
\end{equation}
Using these identities and the idempotency of the projectors, we obtain
\begin{align}
    G_{\boldsymbol{A},\boldsymbol{B}}
    &=
    \delta_{\boldsymbol{\alpha},\boldsymbol{\gamma}}
    \Tr\!\left(
        \Pi_{\boldsymbol{\beta}}
        \V_{\sigma^{-1}\tau}
        \Pi_{\boldsymbol{\delta}}
    \right)
    \nonumber
\end{align}
Orthogonality of the boundary projectors therefore gives
\begin{equation}
    G_{\boldsymbol{A},\boldsymbol{B}}
    =
    \delta_{\boldsymbol{\alpha},\boldsymbol{\gamma}}
    \delta_{\boldsymbol{\beta},\boldsymbol{\delta}}
    \Tr\!\left(
        \Pi_{\boldsymbol{\beta}}
        \V_{\sigma^{-1}\tau}
    \right).
    \label{eq:gram_relative_permutation}
\end{equation}

We are thus left with two diagrams having identical boundary colored words. In this case, for every color $i$, the maps
\begin{equation}
    \sigma_i,\tau_i:
    P_i(\boldsymbol{\beta})
    \longrightarrow
    P_i(\boldsymbol{\alpha})
\end{equation}
are bijections between the same two sets. Their relative composition is therefore a permutation of the positions carrying color $i$,
\begin{equation}
    \sigma_i^{-1}\tau_i
    \in
    S_{n_i}\!\left(P_i(\boldsymbol{\beta})\right),
\end{equation}
and the global relative permutation decomposes color by color as
\begin{equation}
    \left.\sigma^{-1}\tau\right|_{P_i(\boldsymbol{\beta})}
    =
    \sigma_i^{-1}\tau_i.
    \label{eq:relative_permutation_colour_restriction}
\end{equation}
To evaluate the remaining trace, choose an orthonormal basis $\{\ket{x}^{(i)}\}_{x=1}^{d_i}$ of each color sector $\HH_i$. A basis vector in the tensor-product sector selected by $\Pi_{\boldsymbol{\beta}}$ has the form
\begin{equation}
    \ket{x_1,\ldots,x_k},
    \qquad
    \ket{x_r}\in\HH_{\beta_r}.
\end{equation}
Using the convention
\begin{equation}
    \V_\rho
    \ket{x_1,\ldots,x_k}
    =
    \ket{x_{\rho^{-1}(1)},\ldots,x_{\rho^{-1}(k)}},
\end{equation}
the diagonal matrix element entering the trace is nonzero precisely when
\begin{equation}
    x_r
    =
    x_{\rho^{-1}(r)},
    \qquad
    \rho=\sigma^{-1}\tau.
\end{equation}
Thus all basis indices belonging to the same cycle of $\rho$ are identified. Since $\rho$ preserves every set $P_i(\boldsymbol{\beta})$, each cycle has a definite color. A cycle of color $i$ leaves one independent basis index in $\HH_i$ and therefore contributes a factor $d_i$. Consequently,
\begin{equation}
    \Tr\!\left(
        \Pi_{\boldsymbol{\beta}}
        \V_{\sigma^{-1}\tau}
    \right)
    =
    \prod_{i=1}^{m}
    d_i^{\#\operatorname{cycles}(\sigma_i^{-1}\tau_i)}.
\end{equation}
Substitution into Eq.~\eqref{eq:gram_relative_permutation} yields
\begin{equation}
    G_{\boldsymbol{A},\boldsymbol{B}}
    =
    \delta_{\boldsymbol{\alpha},\boldsymbol{\gamma}}
    \delta_{\boldsymbol{\beta},\boldsymbol{\delta}}
    \prod_{i=1}^{m}
    d_i^{\#\operatorname{cycles}(\sigma_i^{-1}\tau_i)},
\end{equation}
which proves the result.
\end{proof}

\section{Proof of Corollary~\ref{cor:gram_decomposition}}
\label{app:C}
\begin{proof}
Fix two boundary colored words $\boldsymbol{\alpha}$ and $\boldsymbol{\beta}$ with common occupation vector $n(\boldsymbol{\alpha})=n(\boldsymbol{\beta})=(n_1,\ldots,n_m)$. For each color $i$, let $P_i(\boldsymbol{\alpha})$ and $P_i(\boldsymbol{\beta})$ denote the corresponding sets of positions on the two boundaries. As shown above, color preservation implies that the restriction of $\sigma$ to color $i$ defines a bijection
\begin{equation}
    \sigma_i:P_i(\boldsymbol{\beta})\longrightarrow P_i(\boldsymbol{\alpha}).
\end{equation}
Choose bijections
\begin{align}
    \xi_{\boldsymbol{\alpha},i}:\{1,\ldots,n_i\}\longrightarrow P_i(\boldsymbol{\alpha}),
    \\
    \xi_{\boldsymbol{\beta},i}:\{1,\ldots,n_i\}\longrightarrow P_i(\boldsymbol{\beta}),
\end{align}
and define
\begin{equation}
    \widehat{\sigma}_i
    :=
    \xi_{\boldsymbol{\alpha},i}^{-1}\ \sigma_i\ \xi_{\boldsymbol{\beta},i}.
\end{equation}
Since $\widehat{\sigma}_i$ is a bijection from $\{1,\ldots,n_i\}$ to itself, it defines a permutation $\widehat{\sigma}_i\in S_{n_i}$. Hence, for fixed boundary words $(\boldsymbol{\alpha},\boldsymbol{\beta})$, every color-preserving diagram is parametrized by a tuple
\begin{equation}
    (\widehat{\sigma}_1,\ldots,\widehat{\sigma}_m)
    \in
    S_{n_1}\times\cdots\times S_{n_m}.
\end{equation}

Consider now two color-preserving diagrams with the same boundary words, associated with permutations $\sigma$ and $\tau$. For each color $i$, define analogously
\begin{equation}
    \widehat{\tau}_i
    :=
    \xi_{\boldsymbol{\alpha},i}^{-1} \ \tau_i\ \xi_{\boldsymbol{\beta},i}
    \in S_{n_i}.
\end{equation}
Their relative permutations satisfy
\begin{equation}
    \widehat{\sigma}_i^{-1}\widehat{\tau}_i
    =
    \xi_{\boldsymbol{\beta},i}^{-1} \ 
    \left(\sigma_i^{-1}\tau_i\right)\ 
    \xi_{\boldsymbol{\beta},i}.
\end{equation}
Therefore $\widehat{\sigma}_i^{-1}\widehat{\tau}_i$ and $\sigma_i^{-1}\tau_i$ differ only by a relabelling of the positions and have the same cycle structure,
\begin{equation}
    \#\operatorname{cycles}\!\left(\widehat{\sigma}_i^{-1}\widehat{\tau}_i\right)
    =
    \#\operatorname{cycles}\!\left(\sigma_i^{-1}\tau_i\right).
\end{equation}

Consider the restriction of the Gram matrix to the block associated with fixed boundary words $(\boldsymbol{\alpha},\boldsymbol{\beta})$. Thus, for
\begin{equation}
    \boldsymbol{A}=(\boldsymbol{\alpha},\boldsymbol{\beta};\sigma),
    \qquad
    \boldsymbol{B}=(\boldsymbol{\alpha},\boldsymbol{\beta};\tau),
\end{equation}
the boundary Kronecker deltas in Proposition~\ref{prop:block_commutant_gram} are identically equal to one.

Using Proposition~\ref{prop:block_commutant_gram} and the equality of the cycle structures under relabelling, the Gram matrix elements within this fixed-boundary block are
\begin{equation}
    G^{(\boldsymbol{\alpha},\boldsymbol{\beta})}_{\boldsymbol{A},\boldsymbol{B}}
    =
    \prod_{i=1}^{m}
    d_i^{\#\operatorname{cycles}(\widehat{\sigma}_i^{-1}\widehat{\tau}_i)}
    =
    \prod_{i=1}^{m}
    G(n_i,d_i)_{\widehat{\sigma}_i,\widehat{\tau}_i}.
\end{equation}
where
\begin{equation}
    G(n_i,d_i)_{\widehat{\sigma}_i,\widehat{\tau}_i}
    =
    d_i^{\#\operatorname{cycles}(\widehat{\sigma}_i^{-1}\widehat{\tau}_i)}
\end{equation}
is the ordinary unitary Gram matrix on $n_i$ replicas in dimension $d_i$. Since the matrix elements factorize independently over the colors, the corresponding Gram block is
\begin{equation}
    G^{(\boldsymbol{\alpha},\boldsymbol{\beta})}
    \cong
    \bigotimes_{i=1}^{m}G(n_i,d_i).
\end{equation}

It remains to account for the different boundary words. For a fixed occupation vector $n$, the number of colored words with occupations $(n_1,\ldots,n_m)$ is
\begin{equation}
    N_{n}
    =
    \frac{k!}{\prod_{i=1}^{m}n_i!}.
\end{equation}
The upper and lower boundary words can be chosen independently, giving $N_{n}^2$ boundary pairs with the same occupation vector. Proposition~\ref{prop:block_commutant_gram} further shows that Gram matrix elements vanish between diagrams associated with distinct pairs of boundary words. Therefore the Gram matrix is block diagonal in $(\boldsymbol{\alpha},\boldsymbol{\beta})$, and each of the $N_{n}^2$  blocks associated with a fixed occupation vector $n$ is isomorphic to $\bigotimes_i G(n_i,d_i)$. This yields
\begin{equation}
    G
    \cong
    \bigoplus_{\substack{n,\\ |n|=k}}
    \left[
        \bigotimes_{i=1}^{m} G(n_i,d_i)
    \right]^{\oplus N_{n}^2}.
\end{equation}
\end{proof}
We further show that the dimensions are coherent. 
For $d_i\geq k$, the ordinary Gram matrix $G(n_i,d_i)$ is indexed by permutations in $S_{n_i}$ and therefore has $n_i!$ rows and columns. Hence, for a fixed occupation vector $n$, the tensor-product block has dimension $\prod_{i=1}^m n_i!$,
and, taking into account its multiplicity $N_{n}^2$, contributes 
\begin{equation} 
N_{n}^2\prod_{i=1}^m n_i! = \frac{(k!)^2}{\prod_{i=1}^m n_i!} = k!\binom{k}{n_1,\ldots,n_m} 
\end{equation} rows and columns. Summing over all occupation vectors with $|n|=k$ and using the multinomial theorem gives 
\begin{equation} 
k! \sum_{\substack{n_1+\cdots+n_m=k\\ n_i\geq0}} \binom{k}{n_1,\ldots,n_m} = k!m^k. \end{equation} 
Therefore the full Gram matrix has $k!m^k$ rows and columns.\\

Let us recall the example introduced in Sec.~\ref{sec_Gram}, where we considered the color-restricted maps $\sigma_i$ and $\tau_i$. Here, we represent these maps as ordinary permutations by relabeling the corresponding position sets. This allows us to make the construction fully explicit and to illustrate how the individual color sectors combine in the corresponding Gram-matrix block.

For $k=5$ and $m\geq 3$, consider the colored words
\begin{equation}
\begin{aligned}
\boldsymbol{\alpha}
&=
\bigl(
\cpdiagdot{cpdiag color 2},
\cpdiagdot{cpdiag color 1},
\cpdiagdot{cpdiag color 3},
\cpdiagdot{cpdiag color 2},
\cpdiagdot{cpdiag color 1}
\bigr),
\\
\boldsymbol{\beta}
&=
\bigl(
\cpdiagdot{cpdiag color 3},
\cpdiagdot{cpdiag color 1},
\cpdiagdot{cpdiag color 2},
\cpdiagdot{cpdiag color 1},
\cpdiagdot{cpdiag color 2}
\bigr).
\end{aligned}
\label{app:eq:alpha_beta_restricted_example}
\end{equation}
Both words have occupation vector
\begin{equation}
n(\boldsymbol{\alpha})
=
n(\boldsymbol{\beta})
=
\bigl(
        \cpnum{cpdiag color 2}{2},
        \cpnum{cpdiag color 1}{2},
        \cpnum{cpdiag color 3}{1},
        \cpnum{cpdiag color 4}{0},
        \cdots,
        \cpnum{cpdiag color 5}{0}
    \bigr).
\end{equation}
The corresponding position sets are
\begin{align}
P_{\cpdiagdot{cpdiag color 1}}(\boldsymbol{\alpha})
    &=\{2,5\},
&
P_{\cpdiagdot{cpdiag color 1}}(\boldsymbol{\beta})
    &=\{2,4\},
\nonumber\\
P_{\cpdiagdot{cpdiag color 2}}(\boldsymbol{\alpha})
    &=\{1,4\},
&
P_{\cpdiagdot{cpdiag color 2}}(\boldsymbol{\beta})
    &=\{3,5\},
\nonumber\\
P_{\cpdiagdot{cpdiag color 3}}(\boldsymbol{\alpha})
    &=\{3\},
&
P_{\cpdiagdot{cpdiag color 3}}(\boldsymbol{\beta})
    &=\{1\}.
\label{app:eq:position_sets_example}
\end{align}
We consider two color-preserving permutations,
\begin{equation}
    \sigma=(13)(254),
    \qquad
    \tau=(1345)(2),
\label{app:eq:sigma_tau_example}
\end{equation}
both mapping the boundary word $\boldsymbol{\beta}$ to
$\boldsymbol{\alpha}$.

Their restrictions to each color sector are bijections
\begin{equation}
    \sigma_i,\tau_i:
    P_i(\boldsymbol{\beta})
    \longrightarrow
    P_i(\boldsymbol{\alpha})\ .
\end{equation}
For the first color, one obtains
\begin{align}
\sigma_{\cpdiagdot{cpdiag color 1}}:
\{2,4\}&\longrightarrow\{2,5\},
&
2&\mapsto5,
&
4&\mapsto2,
\nonumber\\
\tau_{\cpdiagdot{cpdiag color 1}}:
\{2,4\}&\longrightarrow\{2,5\},
&
2&\mapsto2,
&
4&\mapsto5.
\label{app:eq:restricted_blue}
\end{align}
For the second color,
\begin{align}
\sigma_{\cpdiagdot{cpdiag color 2}}:
\{3,5\}&\longrightarrow\{1,4\},
&
3&\mapsto1,
&
5&\mapsto4,
\nonumber\\
\tau_{\cpdiagdot{cpdiag color 2}}:
\{3,5\}&\longrightarrow\{1,4\},
&
3&\mapsto4,
&
5&\mapsto1,
\label{app:eq:restricted_red}
\end{align}
while for the third color,
\begin{align}
\sigma_{\cpdiagdot{cpdiag color 3}}:
\{1\}&\longrightarrow\{3\},
&
1&\mapsto3,
\nonumber\\
\tau_{\cpdiagdot{cpdiag color 3}}:
\{1\}&\longrightarrow\{3\},
&
1&\mapsto3.
\label{app:eq:restricted_green}
\end{align}

Although $\sigma_i$ and $\tau_i$ are naturally bijections between
subsets of replica positions, they can equivalently be represented as
ordinary permutations after canonically relabeling these subsets.
Explicitly, in the present example, we consider the order-preserving bijections
\begin{align}
\xi_{\cpdiagdot{cpdiag color 1},\boldsymbol{\alpha}}
    &: \{1,2\}\longrightarrow\{2,5\},
&
1&\mapsto2,\quad 2\mapsto5,
\nonumber\\
\xi_{\cpdiagdot{cpdiag color 1},\boldsymbol{\beta}}
    &: \{1,2\}\longrightarrow\{2,4\},
&
1&\mapsto2,\quad 2\mapsto4,
\\[1mm]
\xi_{\cpdiagdot{cpdiag color 2},\boldsymbol{\alpha}}
    &: \{1,2\}\longrightarrow\{1,4\},
&
1&\mapsto1,\quad 2\mapsto4,
\nonumber\\
\xi_{\cpdiagdot{cpdiag color 2},\boldsymbol{\beta}}
    &: \{1,2\}\longrightarrow\{3,5\},
&
1&\mapsto3,\quad 2\mapsto5,
\\[1mm]
\xi_{\cpdiagdot{cpdiag color 3},\boldsymbol{\alpha}}
    &: \{1\}\longrightarrow\{3\},
&
1&\mapsto3,
\nonumber\\
\xi_{\cpdiagdot{cpdiag color 3},\boldsymbol{\beta}}
    &: \{1\}\longrightarrow\{1\},
&
1&\mapsto1.
\end{align}

We can therefore associate with each restricted map an ordinary
permutation
\begin{equation}
\widehat{\sigma}_i
:=
\xi_{i,\boldsymbol{\alpha}}^{-1}
\sigma_i
\xi_{i,\boldsymbol{\beta}}\ ,
\qquad
\widehat{\tau}_i
:=
\xi_{i,\boldsymbol{\alpha}}^{-1}
\tau_i
\xi_{i,\boldsymbol{\beta}}\ ,
\label{app:eq:hatted_permutations}
\end{equation}
with
$
\widehat{\sigma}_i,\widehat{\tau}_i\in S_{n_i}
$.

The relabeled maps are
\begin{align}
&\widehat{\sigma}_{\cpdiagdot{cpdiag color 1}}=s,
\qquad
\widehat{\tau}_{\cpdiagdot{cpdiag color 1}}=e. \\
&\widehat{\sigma}_{\cpdiagdot{cpdiag color 2}}=e,
\qquad
\widehat{\tau}_{\cpdiagdot{cpdiag color 2}}=s. \\
&\widehat{\sigma}_{\cpdiagdot{cpdiag color 3}}
=
\widehat{\tau}_{\cpdiagdot{cpdiag color 3}}
=
e.
\end{align}

The relative permutations entering the Gram matrix are obtained by
composing the restricted maps,
\begin{equation}
\sigma_i^{-1}\tau_i:
P_i(\boldsymbol{\beta})
\longrightarrow
P_i(\boldsymbol{\beta})\ .
\end{equation}
Under the above relabeling, these are represented by
\begin{equation}
\widehat{\sigma}_i^{-1}\widehat{\tau}_i
=
\xi_{i,\boldsymbol{\beta}}^{-1}
\left(
\sigma_i^{-1}\tau_i \right)
\xi_{i,\boldsymbol{\beta}}\ .
\label{app:eq:relative_relabeling}
\end{equation}
Thus,
\begin{align}
\widehat{\sigma}_{\cpdiagdot{cpdiag color 1}}^{-1}
\widehat{\tau}_{\cpdiagdot{cpdiag color 1}}
=s\ ,
\qquad 
\widehat{\sigma}_{\cpdiagdot{cpdiag color 2}}^{-1}
\widehat{\tau}_{\cpdiagdot{cpdiag color 2}}
=s\ ,
\qquad 
\widehat{\sigma}_{\cpdiagdot{cpdiag color 3}}^{-1}
\widehat{\tau}_{\cpdiagdot{cpdiag color 3}}
=e\ .
\label{app:eq:relative_permutations_example}
\end{align}

Equivalently, before relabeling, the relative maps act directly on the
original replica positions as
\begin{align}
\sigma_{\cpdiagdot{cpdiag color 1}}^{-1}
\tau_{\cpdiagdot{cpdiag color 1}}
&:
\begin{cases}
2\mapsto4,\\
4\mapsto2,
\end{cases}
&
\sigma_{\cpdiagdot{cpdiag color 2}}^{-1}
\tau_{\cpdiagdot{cpdiag color 2}}
&:
\begin{cases}
3\mapsto5,\\
5\mapsto3,
\end{cases}
\nonumber\\
\sigma_{\cpdiagdot{cpdiag color 3}}^{-1}
\tau_{\cpdiagdot{cpdiag color 3}}
&:
1\mapsto1\ .
\end{align}

Since conjugation by the relabeling maps $\xi_{i,\boldsymbol{\beta}}$
does not change the cycle structure, one has
\begin{equation}
\#\operatorname{cycles}(\sigma_i^{-1}\tau_i)
=
\#\operatorname{cycles}
\left(
\widehat{\sigma}_i^{-1}\widehat{\tau}_i
\right)\ .
\label{app:eq:cycle_relabeling_invariance}
\end{equation}

In the present example, the boundary words
$\boldsymbol{\alpha}$ and $\boldsymbol{\beta}$ are fixed, and we have selected two particular color-preserving permutations $\sigma$ and $\tau$. We are therefore computing a single matrix element of the Gram block associated with these boundary words. For this pair,
\begin{equation}
G_{(\boldsymbol{\alpha},\boldsymbol{\beta};\sigma),
  (\boldsymbol{\alpha},\boldsymbol{\beta};\tau)}
=
\prod_{i=1}^{m}
d_i^{\#\operatorname{cycles}(\sigma_i^{-1}\tau_i)}
=
d_{\cpdiagdot{cpdiag color 1}}\
d_{\cpdiagdot{cpdiag color 2}}\
d_{\cpdiagdot{cpdiag color 3}}\ .
\label{eq:app_example_gram_element}
\end{equation}

If instead one keeps the boundary words $\boldsymbol{\alpha},\boldsymbol{\beta}$ fixed and allows $\sigma$ and $\tau$ to run over all compatible color-preserving permutations, one obtains the full Gram block associated with $(\boldsymbol{\alpha},\boldsymbol{\beta})$. After the order-preserving
relabeling introduced above, the compatible permutations are
parametrized independently by
\begin{equation}
S_{n_1}\times\cdots\times S_{n_m} \ ,
\end{equation}
and the corresponding Gram block factorizes over the color sectors.
Finally varying the boundary words reconstructs the full Gram matrix: different pairs of boundary words with the same occupation numbers give identical blocks, while varying the occupation numbers generates the different block types.

\section{Proof of Lemma \ref{lemma:pure_state_block_weingarten}}
\begin{proof}
Let $\ket{\phi_i(U_i)}:=U_i\ket{\eta_i} \in\HH_i$, then
\begin{equation}
    \ket{\psi(U)}^{\otimes k}
    =
    \sum_{\boldsymbol{\alpha}\in\{1,\dots,m\}^k}
    \sqrt{p_{\boldsymbol{\alpha}}}\,
    \ket{\phi_{\alpha_1}(U_{\alpha_1})}
    \otimes\cdots\otimes
    \ket{\phi_{\alpha_k}(U_{\alpha_k})},
\end{equation}
where $p_{\boldsymbol{\alpha}}:=\prod_{r=1}^k p_{\alpha_r}$.
Hence
\begin{equation}
    \left(\proj{\psi(U)}\right)^{\otimes k}
    =
    \sum_{\boldsymbol{\alpha},\boldsymbol{\beta}}
    \sqrt{
        p_{\boldsymbol{\alpha}}
        p_{\boldsymbol{\beta}}
    }\,
    \ket{\Phi_{\boldsymbol{\alpha}}(U)}
    \bra{\Phi_{\boldsymbol{\beta}}(U)}.
\end{equation}
where
$\ket{\Phi_{\boldsymbol{\alpha}}(U)}
=
\bigotimes_{r=1}^k
\ket{\phi_{\alpha_r}(U_{\alpha_r})}$.
The Haar average of the term associated with
$(\boldsymbol{\alpha},\boldsymbol{\beta})$
vanishes unless
$n(\boldsymbol{\alpha})=n(\boldsymbol{\beta})$.
For $n(\boldsymbol{\alpha})=n(\boldsymbol{\beta})=n$, one has
$\sqrt{p_{\boldsymbol{\alpha}}p_{\boldsymbol{\beta}}}
=\prod_{i=1}^m p_i^{n_i}$.
Applying the reorderings $T_{\boldsymbol{\alpha}}$ and
$T_{\boldsymbol{\beta}}$, the Haar average then factorizes over the colors:
\begin{equation}
\begin{aligned}
    &
    T_{\boldsymbol{\alpha}}
    \left[
        \int dU\,
        \Pi_{\boldsymbol{\alpha}}
        \left(\proj{\psi(U)}\right)^{\otimes k}
        \Pi_{\boldsymbol{\beta}}
    \right]
    T_{\boldsymbol{\beta}}^{-1}
    \\
    &\qquad =
    \prod_{i=1}^m p_i^{n_i}
    \bigotimes_{i=1}^m
    \int_{\mathcal{U}(\HH_i)}dU_i\,
    \left(
        U_i\proj{\eta_i} U_i^\dagger
    \right)^{\otimes n_i}.
\end{aligned}
\label{eq:pure_state_reordered_average}
\end{equation}
For a Haar-random pure state in a $d_i$-dimensional Hilbert space,
\begin{equation}
    \int_{\mathcal{U}(\HH_i)}dU_i\,
    \left(
        U_i\proj{\eta_i}U_i^\dagger
    \right)^{\otimes n_i}
    =
    \frac{1}{(d_i)_{n_i}}
    \sum_{\tau_i\in S_{n_i}}\V_{\tau_i}^{(i)}.
    \label{eq:haar_state_moment_single_sector}
\end{equation}
Therefore
\begin{equation}
    M_k(\boldsymbol p)
    =
    \sum_{\boldsymbol A\in\mathcal{C}_k}
    \frac{
        p^{n(\boldsymbol A)}
    }{
        \prod_{i=1}^m (d_i)_{n_i(\boldsymbol A)}
    }
    D_{\boldsymbol A}.
\end{equation}
\end{proof}

\section{Proof of Corollary \ref{cor:Averaging_weights}}
\begin{proof}
Since 
\begin{equation}
    \E_{\boldsymbol p}\left[
        \prod_{i=1}^m p_i^{n_i}
    \right]
    =
    \frac{\prod_{i=1}^m(d_i)_{n_i}}{(d)_k},
    \quad
    |n|=k,
\end{equation}
the proof follows directly
\end{proof}

\section{Second-Rényi entropy in a fixed-charge sector}
\label{app:fixed_charge_renyi2}

We derive Eq.~\eqref{eq:fixed_charge_average_purity}. Let
$\Pi_q$ be the projector onto the fixed-charge sector $\HH_q$,
\begin{equation}
    \Pi_q
    =
    \sum_{a\in\mathcal I_q}
    \Pi_A^{(a)}\otimes \Pi_B^{(q-a)} ,
    \label{eq:app_fixed_charge_projector}
\end{equation}
where $\Pi_A^{(a)}$ projects onto $\HH_A^{(a)}$ and
$\Pi_B^{(q-a)}$ projects onto $\HH_B^{(q-a)}$.
For a Haar-random pure state in $\HH_q$, the second moment is
\begin{equation}
    \E_q[\psi_q^{\otimes 2}]
    =
    \frac{
    \Pi_q^{\otimes 2}
    \left(\id+\V_A\V_B\right)
    \Pi_q^{\otimes 2}
    }{
    d_q(d_q+1)
    },
    \label{eq:app_fixed_charge_second_moment}
\end{equation}
where $\V_A$ swaps the two copies of subsystem $A$, and $\V_B$ swaps the two
copies of subsystem $B$. Since $\Pi_q^{\otimes 2}$ commutes with the total swap
$\V_A\V_B$, we will omit redundant projectors in the trace manipulations below.

The purity of the reduced state can be written with the swap trick as
\begin{equation}
    \Tr \rho_A^2
    =
    \Tr\left[
    \V_A\,\psi_q^{\otimes 2}
    \right].
    \label{eq:app_swap_trick_purity}
\end{equation}
Combining Eqs.~\eqref{eq:app_fixed_charge_second_moment} and
\eqref{eq:app_swap_trick_purity} gives
\begin{equation}
\begin{split}
    \E_q[\Tr \rho_A^2]
    &=
    \frac{
    \Tr\left[\V_A\Pi_q^{\otimes 2}\right]
    +
    \Tr\left[\V_A\V_A\V_B\Pi_q^{\otimes 2}\right]
    }{
    d_q(d_q+1)
    }\\
    &= \frac{
    \Tr\left[\V_A\Pi_q^{\otimes 2}\right]
    +
    \Tr\left[\V_B\Pi_q^{\otimes 2}\right]
    }{
    d_q(d_q+1)
    } .
\end{split}
\end{equation}
where we used $\V_A^2=\id$.
We now evaluate the two traces. Using
Eq.~\eqref{eq:app_fixed_charge_projector}, 
we are left with 
\begin{equation}
    \Tr_A\left[
    \V_A
    \left(
    \Pi_A^{(a)}\otimes \Pi_A^{(a')}
    \right)
    \right]
    =
    \delta_{a,a'}d_A(a),
\end{equation}
and 
\begin{equation}
    \Tr_B\left[
    \Pi_B^{(q-a)}\otimes \Pi_B^{(q-a')}
    \right]
    =
    d_B(q-a)d_B(q-a').
\end{equation}
Therefore
\begin{equation}
    \Tr\left[\V_A\Pi_q^{\otimes 2}\right]
    =
    \sum_{a\in\mathcal I_q}
    d_A(a)d_B(q-a)^2 ,
    \label{eq:app_trace_swap_A}
\end{equation}
and similarly for $b$.
In the end, we are left with
\begin{equation}
    \E_q[\Tr \rho_A^2]
    =
    \frac{
    \sum_{a\in\mathcal I_q}
    \left[
    d_A(a)d_B(q-a)^2
    +
    d_A(a)^2d_B(q-a)
    \right]
    }{
    d_q(d_q+1)
    } .
\end{equation}
The same result follows from the Bianchi--Donà parametrization. Write
\begin{equation}
    \ket{\psi_q}
    =
    \sum_{a\in\mathcal I_q}
    \sqrt{p_a}\,\ket{\phi_a},
    \qquad
    \ket{\phi_a}\in
    \HH_A^{(a)}\otimes\HH_B^{(q-a)} ,
\end{equation}
with $\ket{\phi_a}$ independently Haar-distributed in each block and
$\boldsymbol p$ Dirichlet-distributed with parameters
$D_a=d_A(a)d_B(q-a)$. Since different $a$ sectors are orthogonal on $A$,
\begin{equation}
    \rho_A
    =
    \bigoplus_{a\in\mathcal I_q}
    p_a \rho_A^{(a)},
    \qquad
    \rho_A^{(a)}
    =
    \Tr_B \ket{\phi_a}\bra{\phi_a}.
\end{equation}
Hence
\begin{equation}
    \Tr \rho_A^2
    =
    \sum_{a\in\mathcal I_q}
    p_a^2\Tr\left[(\rho_A^{(a)})^2\right].
\end{equation}
For a Haar-random state in
$\HH_A^{(a)}\otimes\HH_B^{(q-a)}$,
\begin{equation}
    \E_{\phi_a}
    \Tr\left[(\rho_A^{(a)})^2\right]
    =
    \frac{d_A(a)+d_B(q-a)}{D_a+1}.
\end{equation}
For the Dirichlet distribution with parameters $D_a$,
\begin{equation}
    \E_{\boldsymbol p}[p_a^2]
    =
    \frac{D_a(D_a+1)}{d_q(d_q+1)} .
\end{equation}
Combining the last two equations gives
\begin{equation}
    \E_q[\Tr \rho_A^2]
    =
    \sum_{a\in\mathcal I_q}
    \frac{D_a(D_a+1)}{d_q(d_q+1)}
    \frac{d_A(a)+d_B(q-a)}{D_a+1},
\end{equation}
which reduces again to Eq.~\eqref{eq:fixed_charge_average_purity}.

\section{ETH ansatz calculations}
\label{app:ETH}
We prove Eq.~\eqref{eq:ETH_ansatz} evaluating explicitely, for $k=1,2$,  
\begin{align}
        &\overline{O^U_{i_1 j_1}\cdots O^U_{i_k j_k}} = \\  &\langle E_{i_1}\otimes\cdots\otimes E_{i_k}| \Phi_{U(\mathcal{H}|\mathbf{M})}^{ (k)}\left(O^{\otimes k}\right) |E_{j_1}\otimes\cdots\otimes E_{j_k}\rangle \ ,
    \notag
\end{align}
through 
\begin{equation}
\begin{split}
    \Phi_{\mathcal{U}(\HH|\textbf{M})}^{ (k)}(O^{\otimes k })&= \sum_{\boldsymbol{A},\boldsymbol{B} \in \mathcal{C}_k} \Wg_{\boldsymbol{A},\boldsymbol{B}} \Tr(D^\dagger_{\boldsymbol{A}} O^{\otimes k}) D_{\boldsymbol{B}}\ , 
\end{split}
\label{app:WC}
\end{equation}
with $D_{\boldsymbol{A}}= D_{\boldsymbol{\alpha},\boldsymbol{\beta}}^\sigma:= \Pi_{\boldsymbol{\alpha}}\V_\sigma\Pi_{\boldsymbol{\beta}} \in \Comm_k(\mathcal{U}(\HH|\textbf{M}))$, $\boldsymbol{A}=(\boldsymbol{\alpha}, \boldsymbol{\beta};\sigma)$.

We highlight, since will be relevant for the calculation, that the projectors on the energy windows as defined in Eq.~\eqref{eq:ETH_energy_projectors} have matrix elements on the energy basis which are 
\begin{equation}
    \langle E_i |\Pi_\alpha| E_j\rangle = \sum_{k:E_k \in E_\alpha^\Delta} \langle E_i|E_k\rangle\langle E_k|E_j\rangle = \delta_{\alpha J} \delta_{ij} \ ,
\end{equation}
having defined $J$ so that $E_j \in E^\Delta_J$. In the same way for tensor products.

\subsection{$k=1$}

We start from the simplest case, $k=1$. Since $S_1=\{e\}$ contains only the identity permutation, the color-preserving condition implies that  $\alpha=\beta$ so that the set of admissible triples reduces to
\begin{equation}
    \mathcal{C}_1=\{(\alpha,\alpha;e):\alpha=1,\ldots,m\}\ .
\end{equation}
The corresponding commutant elements are therefore
\begin{equation}
    D^e_{\alpha,\alpha}=\Pi_\alpha \mathbb{V}_e\Pi_\alpha=\Pi_\alpha\ ,
\end{equation}
and hence
\begin{equation}
    \operatorname{Com}_1\!\left(\mathcal{U}(\mathcal{H}|\mathcal{M})\right)=\operatorname{span}\{\Pi_\alpha:\alpha=1,\ldots,m\}\ .
\end{equation}
In other words, at first order the only operators surviving the block-Haar
average are the $m$ projectors onto the individual color sectors,
corresponding to single-replica color-preserving identity diagrams, which can
carry any of the $m$ possible colors.

At the level of the Gram matrix, we can write explicitly the $k=1$
specialization of the block decomposition derived above. There are $m$
possible occupation patterns,
\begin{equation}
    n
    =
    (0,\ldots,0,1,0,\ldots,0)\ ,
    \qquad
    \alpha=1,\ldots,m\ ,
\end{equation}
each containing a single replica of color $\alpha$. For each occupation
pattern, $N_n=1$, so that there is a single one-dimensional
Gram block
\begin{equation}
    G(1,d_\alpha)=[d_\alpha]\ .
\end{equation}
Consequently, the full Gram matrix decomposes as
\begin{equation}
    G
    \cong
    \bigoplus_{\alpha=1}^{m} G(1,d_\alpha)
    =
    \bigoplus_{\alpha=1}^{m}[d_\alpha]\ ,
\end{equation}
or, equivalently, it is diagonal
\begin{equation}
    G_{\boldsymbol{A},\boldsymbol{B}}
    =
    \delta_{\boldsymbol{A}\boldsymbol{B}}\,d_\alpha \ .
\end{equation}
Its inverse has the same block structure,
\begin{equation}
    \Wg
    \cong
    \bigoplus_{\alpha=1}^{m}
    \left[\frac{1}{d_\alpha}\right]\ ,
    \qquad
    \Wg_{\boldsymbol{A},\boldsymbol{B}}
    =
    \frac{\delta_{\boldsymbol{A}\boldsymbol{B}}}{d_\alpha}\ .
\end{equation}

We can now evaluate explicitly the first-order twirling map from Eq.~\eqref{app:WC}. The sums over color-preserving diagrams reduce to sums over the sector labels,
\begin{align}
    \Phi_{\mathcal{U}(\HH|\textbf{M})}^{(1)}(O)
&=\sum_{\boldsymbol{A},\boldsymbol{B}\in\mathcal{C}_1}\frac{\delta_{\boldsymbol{A}\boldsymbol{B}}}{d_\alpha}\,\Tr(D^\dagger_{\boldsymbol{A}}O)\,D_{\boldsymbol{B}} \\
    &=\sum_{\alpha=1}^{m}\frac{\Tr(\Pi_\alpha O)}{d_\alpha}\Pi_\alpha\ . 
\end{align}

Taking matrix elements in the energy eigenbasis, considering $E_i\in E^\Delta_I$ and $E_j\in E^\Delta_J$,
\begin{align}
\overline{O^U_{ij}}=\sum_{\alpha=1}^{m}\frac{\operatorname{Tr}\!\left(\Pi_{\alpha}O\right)}{d_{\alpha}} \delta_{\alpha J } \delta_{ij}&=\delta_{ij} \operatorname{Tr}\!\left(\hat \Pi_{I}O\right)\\
    &= \delta_{ij} K_1(E_I)\ ,
\end{align} 
as in Eq.~\eqref{eq_ETH1}.

\subsection{$k=2$}

We now consider the second moment, $k=2$. In this case $S_2=\{e,s\}$, where $s=(12)$ exchanges the two replicas. Since $e^{-1}=e$ and $s^{-1}=s$, both permutation operators are Hermitian, $\V_e^\dagger=\V_e$ and $\V_s^\dagger=\V_s$. The color-preserving diagrams naturally separate into two classes according to their color occupation pattern: monochromatic diagrams, in which the two replicas carry the same color, and mixed-color diagrams, in which they carry two distinct colors. This decomposition is precisely reflected in the block structure of the Gram matrix, and in the second-order twirling map~\eqref{app:WC} which is separated into monochromatic and mixed-color contributions,
\begin{equation}
    \Phi_{\mathcal{U}(\HH|\textbf{M})}^{(2)}(O^{\otimes2})
    =
    \Phi_{\mathrm{mono}}^{(2)}(O^{\otimes2})
    +
    \Phi_{\mathrm{mix}}^{(2)}(O^{\otimes2}).
\end{equation}

\paragraph{Monochromatic sector.}
The monochromatic sector corresponds to occupation patterns with two replicas
of the same color,
\begin{equation}
    n
    =
    (0,\ldots,0,2,0,\ldots,0),
    \qquad
    \alpha=1,\ldots,m.
\end{equation}
There are therefore $m$ such occupation patterns.
For each of them
$N_n=1$, while the two permutations $e,s\in S_2$ give the
two color-preserving diagrams, indicated here with the shorthand notation
\begin{equation}
    \boldsymbol{A}_{\alpha,e}
    =
    ((\alpha,\alpha),(\alpha,\alpha);e),
    \quad
    \boldsymbol{A}_{\alpha,s}
    =
    ((\alpha,\alpha),(\alpha,\alpha);s),
\end{equation}
with corresponding commutant elements
\begin{equation}
    D_{\boldsymbol{A}_{\alpha,e}}
    =
    \Pi_\alpha^{\otimes2},
    \qquad
    D_{\boldsymbol{A}_{\alpha,s}}
    =
    \Pi_\alpha^{\otimes2}\V_s \Pi^{\otimes 2}_\alpha.
\end{equation}

Each occupation pattern therefore gives one ordinary $2\times2$ unitary
Gram block,
\begin{equation}
    G(2,d_\alpha)
    =
    \begin{pmatrix}
        d_\alpha^2 & d_\alpha\\
        d_\alpha & d_\alpha^2
    \end{pmatrix}\ ,
\end{equation}
so that the monochromatic part of the Gram matrix consists of $m$
two-dimensional blocks,
\begin{equation}
    G_{\mathrm{mono}}
    \cong
    \bigoplus_{\alpha=1}^{m}G(2,d_\alpha)\ .
\end{equation}

The inverse of each block gives the corresponding Weingarten block,
\begin{equation}
    \Wg(2,d_\alpha)
    =
    G(2,d_\alpha)^{-1}
    =
    \frac{1}{d_\alpha^2-1}
    \begin{pmatrix}
        1 & -d_\alpha^{-1}\\
        -d_\alpha^{-1} & 1
    \end{pmatrix}\ ,
\end{equation}
and therefore
\begin{equation}
    \Wg_{\mathrm{mono}}
    \cong
    \bigoplus_{\alpha=1}^{m}\Wg(2,d_\alpha)\ .
\end{equation}
Equivalently,
\begin{align}
    \Wg_{\boldsymbol{A}_{\alpha,e},\boldsymbol{A}_{\alpha,e}}
    &=
    \Wg_{\boldsymbol{A}_{\alpha,s},\boldsymbol{A}_{\alpha,s}}
    =
    \frac{1}{d_\alpha^2-1} \ ,
    \\
    \Wg_{\boldsymbol{A}_{\alpha,e},\boldsymbol{A}_{\alpha,s}}
    &=
    \Wg_{\boldsymbol{A}_{\alpha,s},\boldsymbol{A}_{\alpha,e}}
    =
    -\frac{1}{d_\alpha(d_\alpha^2-1)} \ .
\end{align}

Using the Weingarten block above in Eq.~\eqref{app:WC}, the monochromatic
contribution to the second-order channel reads
\begin{align}
     &\Phi_{\mathrm{mono}}^{(2)}(O^{\otimes2})=
    \\
    & =
    \sum_{\alpha=1}^{m}
    \sum_{\sigma,\tau\in\{e,s\}}
    \Wg_{\boldsymbol{A}_{\alpha,\sigma},
        \boldsymbol{A}_{\alpha,\tau}}
    \Tr\!\left(
        D_{\boldsymbol{A}_{\alpha,\sigma}}^\dagger O^{\otimes2}
    \right)
    D_{\boldsymbol{A}_{\alpha,\tau}}\ . \notag
\end{align}
The two contractions
entering the monochromatic channel are
\begin{align}
    \Tr\!\left(
        D_{\boldsymbol{A}_{\alpha,e}}^\dagger O^{\otimes2}
    \right)
    &=
    \Tr(\Pi_\alpha O)^2\ ,
    \\ 
    \Tr\!\left(
        D_{\boldsymbol{A}_{\alpha,s}}^\dagger O^{\otimes2}
    \right)
    &=
    \Tr(\Pi_\alpha O\Pi_\alpha O)\ .
\end{align}
The monochromatic contribution to the second-order channel is therefore
\begin{align}
    &\Phi_{\mathrm{mono}}^{(2)}(O^{\otimes2}) = 
    \\&=
    \sum_{\alpha=1}^{m}\frac{1}{d_\alpha^2-1}
    \Bigg[
        \left(
            \Tr(\Pi_\alpha O)^2
            -
            \frac{1}{d_\alpha}
            \Tr(\Pi_\alpha O\Pi_\alpha O)
        \right)
        \Pi_\alpha^{\otimes2}
        \nonumber\\
        &\qquad \qquad +
        \left(
            \Tr(\Pi_\alpha O\Pi_\alpha O)
            -
            \frac{1}{d_\alpha}
            \Tr(\Pi_\alpha O)^2
        \right)
        \Pi_\alpha^{\otimes2}\V_s
    \Bigg]. \notag
\end{align}

\paragraph{Mixed-color sector.}
The mixed-color sector corresponds instead to occupation patterns containing
one replica of each of two distinct colors. For every unordered pair
$\{\alpha,\beta\}$, with $\alpha\neq\beta$, the corresponding occupation
pattern has
\begin{equation}
    n_\alpha=n_\beta=1\ ,
    \qquad \quad 
    n_\gamma=0 \ ,
    \quad
    \gamma\neq\alpha,\beta\ .
\end{equation}
There are $\binom{m}{2}$ such occupation patterns. For each of them,
\begin{equation}
    N_{n}=\frac{2!}{1!1!}=2 \ ,
\end{equation}
corresponding to the two boundary colored words
$(\boldsymbol \alpha, \boldsymbol \beta)$ and $(\boldsymbol \beta, \boldsymbol\alpha)$.
Hence, each occupation pattern gives $N_{n}^2=4$ one-dimensional Gram blocks.
The four color-preserving diagrams are
\begin{align}
    \boldsymbol{A}_{\alpha\beta,e}
    &=
    ((\alpha,\beta),(\alpha,\beta);e),
    &
    \boldsymbol{A}_{\beta\alpha,e}
    &=
    ((\beta,\alpha),(\beta,\alpha);e)\ ,
    \\
    \boldsymbol{A}_{\alpha\beta,s}
    &=
    ((\alpha,\beta),(\beta,\alpha);s),
    &
    \boldsymbol{A}_{\beta\alpha,s}
    &=
    ((\beta,\alpha),(\alpha,\beta);s)\ .
\end{align}
The corresponding commutant elements are
\begin{align}
    D_{\boldsymbol{A}_{\alpha\beta,e}}
    &=
    \Pi_\alpha\otimes\Pi_\beta\ ,
    &
    D_{\boldsymbol{A}_{\beta\alpha,e}}
    &=
    \Pi_\beta\otimes\Pi_\alpha\ ,
    \\
    D_{\boldsymbol{A}_{\alpha\beta,s}}
    &=
    (\Pi_\alpha\otimes\Pi_\beta)\V_s\ ,
    &
    D_{\boldsymbol{A}_{\beta\alpha,s}}
    &=
    (\Pi_\beta\otimes\Pi_\alpha)\V_s\ .
\end{align}
In particular, the swap diagrams are related by Hermitian conjugation,
\begin{equation}
    D_{\boldsymbol{A}_{\alpha\beta,s}}^\dagger
    =
    D_{\boldsymbol{A}_{\beta\alpha,s}}\ .
\end{equation}
Since the permutation group associated with each color is $S_1$, each fixed
pair of boundary words admits a unique color-preserving permutation.
Consequently, the four diagrams above belong to four independent
one-dimensional Gram blocks,
\begin{equation}
    G(1,d_\alpha)\otimes G(1,d_\beta)
    =
    [d_\alpha d_\beta]\ .
\end{equation}
Thus, each mixed-color occupation pattern contributes four scalar blocks,
and
\begin{equation}
    G_{\mathrm{mix}}
    \cong
    \bigoplus_{\alpha<\beta}
    [d_\alpha d_\beta]^{\oplus4}\ .
\end{equation}
The inverse of each block gives the corresponding Weingarten coefficient,
\begin{equation}
    \Wg(1,d_\alpha;1,d_\beta)
    =
    \frac{1}{d_\alpha d_\beta}\ ,
\end{equation}
so that
\begin{equation}
    \Wg_{\mathrm{mix}}
    \cong
    \bigoplus_{\alpha<\beta}
    \left[\frac{1}{d_\alpha d_\beta}\right]^{\oplus4}\ .
\end{equation}
Equivalently, for each of the four mixed-color diagrams,
\begin{equation}
    \Wg_{\boldsymbol{A},\boldsymbol{A}}
    =
    \frac{1}{d_\alpha d_\beta}\ ,
\end{equation}
while matrix elements between distinct Gram blocks vanish.

Using the Weingarten blocks above in Eq.~\eqref{app:WC}, the mixed-color
contribution to the second-order channel reads
\begin{align}
    &\Phi_{\mathrm{mix}}^{(2)}(O^{\otimes2}) =
    \\ &= \sum_{\substack{\alpha,\beta=1\\ \alpha\neq\beta}}^{m}
    \sum_{\sigma\in\{e,s\}}
    \Wg_{\boldsymbol{A}_{\alpha\beta,\sigma},
        \boldsymbol{A}_{\alpha\beta,\sigma}}
    \Tr\!\left(
        D_{\boldsymbol{A}_{\alpha\beta,\sigma}}^\dagger O^{\otimes2}
    \right)
    D_{\boldsymbol{A}_{\alpha\beta,\sigma}}\ . \notag
\end{align}
The two contractions
entering the mixed-color channel are
\begin{align}
    \Tr\!\left(
        D_{\boldsymbol{A}_{\alpha\beta,e}}^\dagger O^{\otimes2}
    \right)
    &=
    \Tr(\Pi_\alpha O)\Tr(\Pi_\beta O)\ ,
    \\
    \Tr\!\left(
        D_{\boldsymbol{A}_{\alpha\beta,s}}^\dagger O^{\otimes2}
    \right)
    &=
    \Tr(\Pi_\alpha O\Pi_\beta O)\ .
\end{align}
The mixed-color contribution to the second-order channel is therefore
\begin{align}
    &\Phi_{\mathrm{mix}}^{(2)}(O^{\otimes2}) = 
    \\&=
    \sum_{\substack{\alpha,\beta=1\\ \alpha\neq\beta}}^{m}
    \frac{1}{d_\alpha d_\beta}
    \Big[
        \Tr(\Pi_\alpha O)\Tr(\Pi_\beta O)
        (\Pi_\alpha\otimes\Pi_\beta)
        \nonumber\\
        &\qquad \qquad \qquad \quad+
        \Tr(\Pi_\alpha O\Pi_\beta O)
        D_{\boldsymbol{A}_{\alpha\beta,s}}
    \Big]\ . \notag
\end{align}

Altogether, the monochromatic sector contributes $m$ two-dimensional blocks,
while the mixed-color sector contributes four one-dimensional blocks for each
pair of distinct colors. Hence,
\begin{equation}
    \dim\operatorname{Com}_2\!\left(\mathcal{U}(\HH|\textbf{M})\right)
    =
    2m+4\binom{m}{2}
    =
    2m^2\ ,
\end{equation}
in agreement with
$\dim\operatorname{Com}_k(\mathcal{U}(\HH|\textbf{M}))=k!m^k$.

We now specialize the second-order channel to the energy eigenbasis. For
$E_i\in E_I^\Delta$ and $E_j\in E_J^\Delta$,
\begin{equation}
    \overline{O^U_{ij}O^U_{ji}}
    =
    \langle E_i\otimes E_j|
    \Phi_{\mathcal{U}(\HH|\textbf{M})}^{(2)}(O^{\otimes2})
    |E_j\otimes E_i\rangle \ .
\end{equation}
The cases $I=J$ and $I\neq J$ select the monochromatic and mixed-color
channels, respectively.

For $I\neq J$, only the swap contribution of the mixed-color channel is
non-vanishing, giving
\begin{equation}
    \overline{O^U_{ij}O^U_{ji}}
    =
    \frac{\Tr(\Pi_I O\Pi_J O)}{d_Id_J}\ ,
    \qquad I\neq J\ .
\end{equation}
For $I=J$ and $i\neq j$, only the swap output of the monochromatic channel
contributes, and
\begin{equation}
    \overline{O^U_{ij}O^U_{ji}}
    =
    \frac{
        \Tr(\Pi_I O\Pi_I O)
        -
        d_I^{-1}\Tr(\Pi_I O)^2
    }{d_I^2-1}\ .
\end{equation}
For $i=j$, both monochromatic outputs contribute,
\begin{equation}
    \overline{O^U_{ii} O^U_{ii}}
    =
    \frac{
        \Tr(\Pi_I O)^2
        +
        \Tr(\Pi_I O\Pi_I O)
    }{d_I(d_I+1)}\ .
\end{equation}

Introducing the normalized projectors
and defining
\begin{align}
    K_2(E_I,E_J)
    &=
    \Tr(\hat{\Pi}_I O\hat{\Pi}_J O)
    -
    \frac{\delta_{IJ}}{d_I}K_1(E_I)^2\ ,
\end{align}
the exact finite-dimensional results can be written as
\begin{align}
    \overline{O^U_{ij}O^U_{ji}}
    &=
    K_2(E_I,E_J)\ ,
    \qquad \quad \qquad   I\neq J\ ,
    \\
    \overline{O^U_{ij}O^U_{ji}}
    &=
    \frac{d_I^2}{d_I^2-1}K_2(E_I,E_I) \ ,
    \qquad I=J,\ i\neq j\ ,
    \\
    \overline{O^U_{ii}O^U_{ii}}
    &=
    K_1(E_I)^2
    +
    \frac{d_I}{d_I+1}K_2(E_I,E_I)\ .
\end{align}

In the asymptotic $d_I\gg1$, the finite-dimensional Weingarten
corrections approach unity. Combining the first- and second-order results,
we obtain
\label{eq:ETH_ansatz_app}
\begin{align}
    \overline{O^U_{ij}O^U_{ji}}
    &\simeq
    K_1(E_I)^2\delta_{ij}
    +
    K_2(E_I,E_J)\ ,
\end{align}
with $E_i\in E_I^\Delta$ and $E_j\in E_J^\Delta$, reproducing
Eq.~\eqref{eq_ETH2}.
\end{document}